\documentclass[11pt]{article}

\usepackage{amsmath,amsfonts,mathrsfs,latexsym,color,array}
\usepackage{fancyhdr}
\usepackage{cite}
\usepackage{url}
\usepackage{ntheorem}
\usepackage{graphicx}
\usepackage{epstopdf}

\newcommand{\spOmega}{N}

\newcommand{\normalX}{X^\perp}
\newcommand{\tangX}{X^\parallel}
\newcommand{\nablah}{\mcD}

\newcommand{\calO}{{\cal O}}

\newcommand{\mdiv}{\operatorname{div}}
\newcommand{\grad}{\operatorname{grad}}

\newcommand{\divh}{\operatorname{div_h}}
\newcommand{\wL}{\tilde C}

\newcommand{\mimath}{i}
\newcommand{\cvl}{\Delta_{\mathbb L}}

\newcommand{\tL}{L}%
\newcommand{\tD}{D}%
\definecolor{bluem}{rgb}{0,0,0.5}

\definecolor{mycolor}{cmyk}{0.5,0.1,0.5,0}
\definecolor{michel}{rgb}{0.5,0.9,0.9}
\definecolor{turquoise}{rgb}{0.25,0.8,0.7}
\definecolor{bluem}{rgb}{0,0,0.5}
\definecolor{MDB}{rgb}{0,0.08,0.45}
\definecolor{MyDarkBlue}{rgb}{0,0.08,0.45}
\definecolor{MLM}{cmyk}{0.1,0.8,0,0.1}
\definecolor{MyLightMagenta}{cmyk}{0.1,0.8,0,0.1}
\definecolor{HP}{rgb}{1,0.09,0.58}

\newcommand{\Ric}{\mathrm{Ric}}

{\catcode `\@=11 \global\let\AddToReset=\@addtoreset}
\AddToReset{equation}{section}

\newcommand{\mcD}{{\mycal D}}

\newcommand{\zX}{\mathring X}
\newcommand{\zK}{\mathring K}

\newtheorem{Theorem} {\sc  Theorem\rm} [section]
\newtheorem{Corollary} [Theorem] {\sc  Corollary\rm}

\newtheorem{lemma} [Theorem] {\sc  Lemma\rm}
\newtheorem{proposition} [Theorem] {\sc  Proposition\rm}
\newtheorem{Proposition} [Theorem] {\sc  Proposition\rm}
\newtheorem{prop} [Theorem] {\sc  Proposition\rm}

\newtheorem{theorem}[Theorem]{\sc  Theorem\rm}

\theorembodyfont{\upshape}

\newtheorem{remark}[Theorem]{\sc Remark\rm}

\newcommand{\beqar}{\begin{deqarr}}
\newcommand{\eeqar}{\end{deqarr}}

\newcommand{\beaa}{\begin{eqnarray*}}
\newcommand{\eeaa}{\end{eqnarray*}}
\newcommand{\bel}[1]{\begin{equation}\label{#1}}
\newcommand{\bea}{\begin{eqnarray}}
\newcommand{\bean}{\begin{eqnarray}\nonumber}
\newcommand{\beal}[1]{\begin{eqnarray}\label{#1}}
\newcommand{\eea}{\end{eqnarray}}
\newcommand{\eeal}[1]{\label{#1}\end{eqnarray}}

\newcommand{\qed}{\hfill $\Box$ \medskip}
\newcommand{\proof}{\noindent {\sc Proof:\ }}
\newcommand{\be}{\begin{equation}}
\newcommand{\eeq}{\end{equation}}
\newcommand{\ee}{\end{equation}}
\newcommand{\beqa}{\begin{eqnarray}}
\newcommand{\eeqa}{\end{eqnarray}}
\newcommand{\beqan}{\begin{eqnarray*}}
\newcommand{\eeqan}{\end{eqnarray*}}
\newcommand{\ba}{\begin{array}}
\newcommand{\ea}{\end{array}}

\newcommand{\cC}{{\cal C}}
\newcommand{\calL}{\mathcal L}
\newcommand{\calC}{\mathcal C}
\newcommand{\calS}{\mathcal S}
\newcommand{\del}{\partial}

\newcommand{\calF}{\mathcal F}
\newcommand{\calP}{\mathcal P}
\newcommand{\calQ}{\mathcal Q}
\newcommand{\calR}{\mathcal R}

\DeclareFontFamily{OT1}{rsfs}{} \DeclareFontShape{OT1}{rsfs}{m}{n}{<-7> rsfs5 <7-10> rsfs7 <10-> rsfs10}{}
\DeclareMathAlphabet{\mycal}{OT1}{rsfs}{m}{n}

\newcommand{\R}{\mathbb R}
\newcommand{\RR}{\mathbb R}
\newcommand{\CC}{\mathbb C}

\newcommand{\N}{\mathbb N}
\newcommand{\bit}{\begin{itemize}}
\newcommand{\eit}{\end{itemize}}

\newcommand{\tg}{{\tilde g}}

\newcommand{\tr}{{\mbox{\rm tr\,}}}

\newcounter{shownewstuffflag}

\newcommand{\startnewstuff}{\ifnum\value{shownewstuffflag}>0\color{blue}\fi}
\newcommand{\finishnewstuff}{\ifnum\value{shownewstuffflag}>0\color{black}\fi}

\newcounter{oldeq}

\newcounter{mnotecount}[section]

\newcommand{\rmnote}[1]{}

\def\beq{\begin{equation}}
\def\eeq{\;. \end{equation}}

\newcommand{\zg}{{\mathring g}}

\newcommand{\HH}{\mathbb H}

\newcommand{\eq}[1]{(\ref{#1})}

\begin{document}
\title{Initial data sets with  ends of cylindrical type:\\ II. The vector constraint equation}

\author{Piotr T. Chru\'sciel\thanks{Supported in part
by the Polish Ministry of Science and Higher Education grant Nr
N N201 372736.}
\\ Institut des Hautes \'Etudes
Scientifiques, Bures-sur-Yvette, \\ and University of Vienna
\\
\\
Rafe Mazzeo \thanks{Supported in part by the NSF Grant DMS-1105050}
\\ Stanford University
\\
\\
Samuel Pocchiola
\\ Universit\'e Paris Sud
}
\date{}

\maketitle\thispagestyle{fancy} \rhead{\bfseries UWThPh-2012-16 \\ \date{March 21, 2012}}

\abstract{We construct solutions of the vacuum vector constraint equations on manifolds with cylindrical ends.}

\section{Introduction} \label{Sintro}
In a companion paper to this one~\cite{PTCMazzeoCylindrical} we have constructed large families of solutions
of the Lichnerowicz equation on manifolds with cylindrical ends. This paper addresses the complementary
problem of constructing solutions of the vacuum vector constraint equation.

Suppose that $(M^n,g)$ is a Riemannian manifold of dimension $n \geq 3$, and $K$ a symmetric $2$-tensor on $M$.
The vacuum constraint equations take the familiar form
\begin{eqnarray}
& R(g)  = 2\Lambda + |K|_g^2 - (\tr_g K)^2  \label{c-scal} \\
& \mbox{div}_g \, K + \nabla \tr_g K = 0;  \label{c-vec}
\end{eqnarray}
where $\Lambda \in \RR$ is the cosmological constant.  These are called, respectively, the scalar and vector constraint
equations and data $(M,g,K)$ which satisfy both are called initial data sets.  If $\tau := \tr_g K$ is constant, then
the conformal method allows one to effectively decouple these equations. We review this below.

As explained in \cite{PTCMazzeoCylindrical}, there are compelling physical reasons for studying
this problem when the initial metric $(M,g)$ is complete with a finite number of cylindrically
bounded ends, possibly accompanied by a finite number of asymptotically hyperbolic or
asymptotically Euclidean or conic ends. That paper initiated the general study of the constraint
equations on manifolds with ends of cyindrical type, but focused exclusively on the Lichnerowicz
equation.  More specifically, we assumed there the existence of some symmetric $2$-tensor $K$,
which need not have constant trace nor be divergence free, and then consider only the problem of
finding solutions to \eqref{c-scal}, for this given $K$, with the same type of asymptotic
geometry as the initial metric $g$.  Only when $K$ is transverse-traceless, i.e.\ has constant trace and
vanishing divergence,  do the solutions found there directly correspond to solutions of both constraint
equations. However, we adopted this slightly more general point of view in hopes that the arguments there
might eventually lead to more general non-CMC solutions of the full constraint equations on this class of manifolds.

The goal of this paper is to show that we may produce large classes of solutions to the vector constraint
equation \eqref{c-vec} with $\tau = \tr_g K$ constant on manifolds with cylindrically bounded ends.
This provides a satisfactory complement to the results of \cite{PTCMazzeoCylindrical} in the CMC setting,
and the two sets of results together establish the existence of solutiosn to the full constraint
equations under reasonably general hypotheses.

We refer to~\cite{PTCMazzeoCylindrical} for much of the terminology, definitions and notation
used below. In the current paper we are primarily interested in initial metrics $(M,g)$, the ends of which are
either asymptotically cylindrical or asymptotically periodic. Recall that this means
that $M$ has a finite number of ends $E_\ell$, $\ell =1, \ldots, N$, such that the restriction of $g$
to each $E_\ell$ is either asymptotic to a product cylindrical metric $dx^2 + \zg_\ell$, where
$E_\ell \cong \RR^+ \times N_\ell$ and $(N_\ell, \zg_\ell)$ is a compact Riemannian manifold, or
else this restriction is asymptotic to a periodic metric $\zg_\ell$ on $\RR \times N$ of period $T_\ell$.
Accordingly, we shall study the vector constraint equation in either of these two settings.

\section{The conformal method and the vector constraint equation}   \label{sSYd}
Fix a complete Riemannian manifold $(M,g)$ and a symmetric $2$-tensor $K$ on $M$.
We begin with a review of how, in the special case that $\tau := \tr_g K$ is constant, to which we
refer hereafter as the CMC case, one may solve the two equations \eqref{c-scal} and \eqref{c-vec} in
sequence rather than simultaneously. This is known as the conformal method.

The first step is to decompose
\[
K = \frac{\tau}{n} g + L,
\]
where $L$ is again a symmetric $2$-tensor.  Let $\phi$ be any positive smooth function on $M$.
Set
\[
\tg := \phi^{\frac{4}{n-2}} g, \quad \mbox{and}\ \ \widetilde{K} ^{ij} :=
\frac{\tau}{n} \tg + \phi^{\frac{-2(n+2)}{n-2}} L = \frac{\tau}{n} \tg + \widetilde{L}.
\]
Then a straightforward computation shows that $(M,g,K)$ satisfies the two constraint equations if
and only if $(M, \tg, \widetilde{K})$ does;  this uses strongly that $\tau$ is constant.
In the CMC case, \eqref{c-vec} reduces to the simpler equation
\bel{conf1b}
\nabla_i L^{ij} = 0 ,
\ee
where $\nabla$ is the Levi-Civita connection for $g$, and part of what we are asserting is that $L$ is divergence free
for $g$ if and only if $\widetilde{L}$ is divergence free for $\tg$. Hence the advantage of the conformal
method is that if we first find a TT tensor $L$ with respect to $g$ and then insert this into the first
constraint equation, then \eqref{c-scal} becomes a semilinear elliptic equation for $\phi$, which is
called the Lichnerowicz equation. We can then (attempt to) solve this, and then, having determined $\phi$,
define the corresponding $\widetilde{K}$ and hence produce an initial data set $(M,\tg, \widetilde{K})$.

As explained in the introduction, we concentrate in this paper entirely on the problem of finding
appropriate TT tensors $L$ on manifold with cylindrical ends, and then appeal to
\cite{PTCMazzeoCylindrical} for the solution of the remaining steps in this procedure.

The method for finding $L$ proceeds as follows. Start with an arbitrary trace-free symmetric tensor field $A^{ij}$;
this will be referred to as the {\em seed field}. Now set
\bel{conf100} \tL^{ij} = A^{ij} + C(X) ^{ij} , \ee
where the operator $C$ which appears here is called the \emph{conformal Killing operator}, and is defined by
\bel{conf102} C(X) ^{ij}:= \tD^i Y^j + \tD^j X^i - \frac 2n\tD_k X^k g^{ij}\;.
\ee
The requirement that $L^{ij}$ be divergence-free can then be written as
\begin{equation}
\label{conf103}
(\cvl Y ) ^{j} =  \tD_i A^{ij},
\end{equation}
where
\begin{equation}
(\cvl X)^j := - \tD_i(\tD^i X^j + \tD^j X^i - \frac 2n \tD_k X^k g^{ij}).
 \label{12III12.1}
\end{equation}
This procedure is commonly attributed to York, see \cite{CBY,BartnikIsenberg}. The operator \eq{12III12.1} is usually called
the \emph{conformal vector Laplacian}.

For later reference, let us rewrite the operators above in an invariant way.  First, for any vector field $X$, write
\[
S(X) = \mbox{Sym}\, \tD X^\sharp = \calL_X g;
\]
this is the Killing operator, which is also the Lie derivative of the metric in the direction $X$ as well as the symmetrization
of the covariant derivative $\tD X^\sharp$, where $X^\sharp$ is the $1$-form $g$-dual to $X$. We also use the convention
that the divergence of $X$ is
\[
\delta X = - \tr^g \tD X;
\]
the choice of sign corresponds to the foraml integration by parts formula $\langle \delta X, f \rangle = \langle X, \tD f\rangle$.
In terms of all of these, we have
\[
C(X) = S(X) + \frac{2}{n} \delta X \, g,\quad \mbox{and}\qquad \cvl X = \delta C(X).
\]
Finally, observe that
\[
\langle C(X), h \rangle = \langle X, C^* X \rangle = \langle X, 2\beta h \rangle,
\]
where $\beta h = \delta + \frac12 \tD\, \tr h$ is the Bianchi operator.  Since $C(X)$ is trace-free, we see that
\begin{equation}
\cvl =  C^* C,
\label{factorvl}
\end{equation}
so in particular, when $M$ is complete, $\cvl$ is self-adjoint and nonnegative.

We have now arrived at the problem which will occupy us for the rest of this paper. Let $(M^n,g)$, $n \geq 3$,
be a complete manifold with cylindrically bounded ends. We wish to determine the solvability of the problem
\begin{equation}
\cvl Y = F,
\label{conf1010}
\end{equation}
under various conditions on the decay of the inhomogeneous term $F$.

\section{Cylindrical ends and elliptic operators}
Fix any end $E = \RR^+_x \times N$ of the manifold $M$ and consider the restriction of the metric $g$ to $E$.
We say that the end is cylindrically bounded if
\[
C_1 ( dx^2 + \hat{g} ) \leq g \leq C_2 (dx^2 + \hat{g}),
\]
where $C_1, C_2 > 0$ and $\hat{g}$ is a metric on $N$.  Amongst these we distinguish two cases of particular interest:
the first is when $g$ is asymptotically cylindrical, which means that $g = dx^2 + \zg + h$ where $\zg$
is a metric on $N$ and $|h|_{dx^2 + \zg} \to 0$ as $x \to \infty$; the other is when $ g$ is asymptotically periodic,
which means that $g = \zg + h$ where $\zg$ is the lift to $\RR \times N$ of a metric on $(\RR/T\mathbb Z) \times N$
and, once again, $|h|_{\zg} \to 0$ at infinity. In \cite{PTCMazzeoCylindrical} we work with a slightly more general
class of conformally asymptotically cylindrical or periodic metrics, but we shall not do this here since
by the discussion in the last section, the extra conformal factor can be transformed away, and hence is irrelevant to
the problem at hand. On the other hand, in \S 5 we generalize the notion of asymptotically periodic slightly
to allow ends which are not diffeomorphic to products $\RR^+ \times N$, but are just $\mathbb Z$ covers
of compact manifolds. We refer to the beginning of that section for a better description.

Unlike our previous paper, it is impossible to study \eqref{conf1010} using barrier methods since
this equation is a system. Therefore we must appeal to more powerful, but more technical,
parametrix methods.  Fortunately, these are very well-developed, particularly for manifolds
with cylindrical ends, and we shall be able to quote standard literature.

The result we shall need is of the following type.  Let $\mathcal L$ be a second order symmetric
elliptic operator acting on a vector bundle $V$ over $M$ endowed with a Hermitian metric.
We assume that $\calL$ is geometric, e.g.\ of the form $\nabla^* \nabla  + R$ where $R$
is a curvature endomorphism, or at least has the same asymptotic structure as the metric $g$.
Both of these things are true when $\mathcal L = \cvl$, but we phrase things in a slightly more
general way for the moment.  We let $\calL$ act on weighted Sobolev spaces $H^k_\delta(M,V)$,
defined by the norm
\begin{equation}
||f||_{k,\delta} = \sum_{j=0}^k \int_M |\nabla^j f|^2 e^{-2\delta x}\, dV_g.
\label{wss}
\end{equation}
Here $e^{-2\delta x}$ actually represents a weight function which equals this
exponential on each end, where $x$ is the linear variable, and equals a constant
on the compact part of $M$.
In addition, $\nabla$ denotes the natural extension of the Levi-Civita connection to a map
\[
(T^*M)^{\otimes i} \otimes V \to (T^*M)^{\otimes (i+1)} \otimes V.
\]

We are interested in determining some range of values of the weight parameter $\delta$
such that the mapping
\begin{equation}
\calL: H^{k+2}_\delta(M, V) \longrightarrow H^k_\delta(M,V)
\label{mapping}
\end{equation}
is Fredholm, or even invertible. We first list a useful result whose proof relies only on
local elliptic regularity and standard duality arguments (compare~\cite{Bartnik86}).

\begin{prop}
The mapping \eqref{mapping} is Fredholm for some value of $\delta$ if and only if the
corresponding mapping with $\delta$ replaced by $-\delta$ is also Fredholm.
In addition, for all such Fredholm values of the weight parameter, \eqref{mapping} is injective,
respectively surjective, if and only if the mapping with $\delta$ replaced by $-\delta$ is surjective, respectively injective.
\label{duality}
\end{prop}

The problem then becomes one of determining the set of values $\delta$ for which \eqref{mapping}
is Fredholm, and then the finer problem of determining for which of these Fredholm values, it is
injective or surjective.

As we have already remarked, there are good general criteria for this in both of the geometric
settings of interest here.  These criteria depend on a set of values $\lambda \in \CC$, called
the indicial roots of $\calL$.  While the abstract definition of these indicial roots is simple enough,
and there are some simple general structural results about them, it can be difficult to say much
about their precise locations in any given situation.  Our main results here address this question.
Namely, we derive a set of constraints on the indicial roots of the conformal vector Laplacian
$\cvl$ on manifolds $(M,g)$ with asymptotically cylindrical or asymptotically periodic ends.
Rather than giving further definitions at this general level, we now specialize to the operator of interest.

In the next two sections, we consider the structure of the conformal vector Laplacian  $\cvl$ on manifolds
of the form $\RR \times N$ with translation invariant metric $dx^2 + \zg$ or else
periodic metric $\zg$. In either setting we define the indicial roots and make a series of
calculations which gives some information about their location.  We then use this information
in \S 6 to describe the resulting global mapping properties of $\cvl$ on weighted Sobolev spaces,
and then describe the existence theorems for TT tensors which we can derive from these.

\section{The conformal vector Laplacian  on product cylinders}
 \label{s17II11.1}
In this section we study the analytic properties of the conformal vector Laplacian  $\cvl$
on the cylinder $M = \RR \times N$ endowed with the product metric
\begin{equation}
g = dx^2 + \zg.
\label{prodmet}
\end{equation}
We first derive a more explicit expression for this operator adapted to this product structure,
which leads to the definition of its indicial roots. The remainder of the section is devoted to
the statement and proof of a result about the location of these roots.

\subsection{A formula for the conformal vector Laplacian }
 \label{subsec:cylindre}
Let $X$ be a vector field on $\RR \times N$. There is a natural decomposition
\bel{Xdecomp}
X=  f\partial_{x} + Y
\ee
where $Y$ is tangent to the $N$ factors. We label these components as $\normalX = f$ and $\tangX = Y$.
Also, denote by $\nablah$ the Levi-Civita connection of $(N,\zg)$ and $\calS(Y) = \calL_Y \zg$ the symmetrization of $\nablah Y^\sharp$.


In the calculations below, we use both invariant notation as well as adapted coordinates $(x,y)$,
where $x \in \RR$ and $\{y^A\}$ is a coordinate system on $N$.  We sometimes use the index
$0$ for $x$ and then assume that $A \geq 1$. First note that $\tD \del_x \equiv 0$,
and if $Y, Y'$ are tangent to the $N$ factors, then so is $\tD_Y Y'$, hence $\tD_Y Y' = \nablah_Y Y'$.
(These statements are equivalent to the vanishing of all Christoffel symbols for $g$ which have $x$ indices,
and the identification of all the remaining ones with the Christoffel symbols for $\zg$.)

We begin by calculating $C(X)$ for $X$ as in \eqref{Xdecomp}.  First, by the remarks above,
\[
\tD X = \nablah Y + \del_x Y \otimes dx + \del_x f \, \del_x \otimes dx,
\]
whence
\beal{Ceqspm}
C(X)_{00} &=& 2\left(1-\frac{1}{n} \right) \del_x f  + \frac{2}{n} \delta^{\zg} Y ,  \\
\label{Ceqspm2}  C(X)_{0A} &=& \partial_A f +\partial_x Y_A,  \\
C(X)_{AB} &=& \calS(Y)_{AB} + \frac 2 n ( \delta^{\zg} Y - \partial_x f )\zg_{AB}, \notag \\
& : = & \tilde{C}(Y) - \frac{2}{n} \del_x f \zg_{AB},
\label{Ceqspm3}
\eea
where, by definition
\begin{equation}
\tilde{C}(Y) = \calS(Y) + \frac{2}{n} \delta^{\zg} Y \zg.
\label{ckoN}
\end{equation}
Note that while this `reduced' conformal Killing operator is an operator on $N$, it is \emph{not} the conformal Killing
operator $C^N$ for $(N,\zg)$ because of the different constant in front of the second term (the correct constant in $C^N$
is $2/(n-1)$).   For later use, we record that
\begin{equation}
\begin{split}
\tilde{C}(Y) & = C^N(Y) - \frac{2}{n(n-1)} \delta^{\zg} Y \zg \\
& \quad \Rightarrow  \tr^{\zg} \tilde{C}(Y) = -\frac{2}{n} \delta^{\zg} Y.
\end{split}
\label{trctilde}
\end{equation}
We then derive that
\begin{eqnarray}
\label{DLz} (\cvl  X)^{\perp } & = & - 2 \left( 1-\frac{1}{n} \right) \partial_{x}^{2} f+ \Delta_{\zg} f +  \left(1-\frac{2}{n} \right)\partial_{x} \delta^{\zg} Y \\
\label{DLA}
(\cvl  X)^{\parallel} & = & -\del_x^2 Y + \widetilde{\cvl} Y + \left(\frac{2}{n}-1 \right) \del_x \nablah  f.
\end{eqnarray}
The operator $\widetilde{\cvl}$ which appears here is given by
\begin{equation}
\label{cvLdef}
\widetilde{\cvl} = \delta^{\zg} \tilde{C}.
\end{equation}
It is once again the reduction to $N$ of the conformal vector Laplacian  on $\RR \times N$.  Since $\tilde{C}(Y)$ is not necessarily
trace-free,
 this operator is {\it not the same} as $\frac{1}{2}\tilde{C}^* \tilde{C}$, nor is it equal to the conformal vector Laplacian
for $(N,\zg)$. However, it is still self-adjoint and nonnegative, as can be seen from the identities
\begin{equation}
\widetilde{\cvl} = \frac{1}{2} \tilde{C}^* \tilde{C} + \frac{2}{n(n-1)} \nablah \delta^{\zg},
\label{facredcvl}
\end{equation}
and the quadratic form version
\begin{equation}
\langle Y, \widetilde{\cvl} Y \rangle = \frac{1}{2} || \tilde{C}(Y)||^2 + \frac{2}{n^2} || \delta^{\zg} Y||^2.
\label{redquadcvl}
\end{equation}

\subsection{Indicial roots}
Any elliptic operator $\calL$ on the cylinder $\RR \times N$ which is translation invariant in the $x$ direction
determines a set of indicial roots $\Lambda(\calL)$, which is a discrete set of complex numbers $\{ \lambda_j\}$
with the property that $|\mbox{Im} \lambda_j| \to \infty$ as $j \to \infty$. These numbers measure the precise rate
of exponential growth or decay of the special `separation of variable' solutions to the equation $\calL u = 0$.

Focusing directly on the conformal vector Laplacian  $\cvl$, define the {\it indicial family}  $I_\lambda(\cvl)$ to be the conjugate
of $\calL$ by the Fourier transform $\calF$, $I_\lambda(\cvl) = \calF \circ \cvl  \circ \calF^{-1}$. This amounts
to replacing $\del_x$ by $i\lambda$. Thus, denoting the Fourier transforms of the various components with hats,
we have $I_\lambda(\cvl)(X) = \hat{w} \del_x + \hat{W}$, where
\begin{equation}
\begin{split}
\hat{w}  &  =  \left(\Delta_{\zg} + 2\left(1 - \frac{1}{n}\right) \lambda^2 \right) \hat{f}
+i \lambda \left(1- \frac{2}{n} \right) \delta^{\zg} \hat{Y}, \\
\hat{W} & = \left(\widetilde{\cvl } + \lambda^2\right) \hat{Y}
+ i  \lambda \left(\frac{2}{n} -1  \right)  \nablah  \hat{f}.
\end{split}
\label{indfam}
\end{equation}


It is immediate from the fact that $\cvl$ itself is elliptic when $n \geq  3$ that $I_{\lambda}(\cvl)$ is elliptic for each $\lambda$, and
determines a Fredholm mapping $H^{k+2}(N) \to H^k(N)$ for any $k$.  Furthermore, it is a polynomial in $\lambda$;
such families are sometimes called operator pencils.  The analytic Fredholm theorem states that either there
exists no value of $\lambda$ for which $I_\lambda(\cvl)$ is invertible, or else it is invertible away from
a discrete set of complex numbers $\Lambda(\cvl) = \{\lambda_j\}$, which is by definition the set of indicial roots of $\cvl$.
Standard elliptic regularity implies that this set is independent of $k$. It is also straightforward to see,
using the semiboundedness of $\Delta_{\zg}$ and $\widetilde{\cvl}$, that any horizontal strip
$\{ a < \Im( \lambda) < b \}$ contains at most a finite number of indicial roots; this shows
that $\Lambda (\cvl) \neq \CC$, and also vindicates the assertion that $|\Im( \lambda)_j| \to \infty$.
This latter statement was proved in \cite{AgmonNirenbergBanach} (see also \cite{MelroseMendoza,RafeEdgeCPDE}).
Finally, since $I_\lambda(\cvl)$ is invertible for some values of $\lambda$, this family has index zero, which
means that $\lambda$ is an indicial root if and only if there exists a nontrivial solution to
$I_\lambda(\cvl)( \hat{f}, \hat{Y}) = (0,0)$, or in other words, it suffices to check injectivity rather than surjectivity.

Note that $I_0(\cvl)(f,0) = (0,0)$ for $f \equiv \mbox{const.}$. This means that $0$ is always an indicial root. It is
not hard to see that $\del_t$ generates the entire nullspace of this operator.

The paper \cite{IMP} calculates the full set of indicial roots for $\cvl$ when $(N,\zg)$ is $S^2$ with its standard metric;
the analogous calculation for $(N,\zg)$ a sphere of any dimension (with its standard metric) is given below in
Appendix~\ref{subsec:indicialsphere}.

\subsection{Mapping properties on product cylinders}
We now explain the significance of indicial roots for the mapping properties of $\cvl$ on weighted Sobolev spaces.
\begin{prop}
The mapping
\begin{equation}
\cvl: e^{\delta x} H^{k+2}( \RR \times N) \longrightarrow e^{\delta x} H^k(\RR \times N)
\label{mapcyl}
\end{equation}
is Fredholm if and only if $\delta \neq \mbox{Im}\, \lambda_j$ where $\lambda_j$ is any indicial root of $\cvl$.
Furthermore, if this map is Fredholm, then it is invertible.
\end{prop}

One direction of this is easy.  If $\delta$ does equal the imaginary part of an indicial root, then
there exists a solution of $\cvl u = 0$ which grows or decays like $e^{\delta x}$ both as $x \to +\infty$
and also as $x \to -\infty$. Because of this asymptotic behaviour, $u$ is right on the border of
lying in $e^{\delta x} L^2$, and it is then easy to define a sequence of compactly supported cutoffs $u_j$
of $u$ with disjoint support which have the property that $||u_j||_{k+2,\delta}  \to \infty$ while
$||\cvl u_j ||_{k,\delta} \leq C$. This shows that \eqref{mapcyl} does not have closed range in this case.

The other implication is not much harder. The simplest proof uses the extension of the
Fourier transform in $x$ to the complex plane. The Plancherel theorem for this extended transform
states that the Fourier transform maps $e^{\delta x} L^2(\RR \times N)$ isometrically
to $L^2( \RR \times N)$, where the first factor $\RR$ is the real part of $\lambda = \xi + i\delta$.
If the line $\mbox{Im}\,\lambda = \delta$ contains no indicial roots, then the inverse of the
indicial family $I_\lambda(\cvl)^{-1}$ has norm which is uniformly bounded along this line, and
\[
f(x,z) \mapsto \hat{f}(\lambda, z) \longmapsto \int_{\mathrm{Im} \lambda = \delta}
 e^{ix\lambda } I_{\lambda}(\cvl)^{-1} \hat{f}(\lambda, z)\, d\lambda := u(x,z),
\]
$z \in N$, provides an inverse for \eqref{mapcyl}.

\subsection{The indicial-root-free region}
Our goal in the remainder of this section is to establish a result which describes a region of the plane, which
depends on the lowest eigenvalue $\lambda_1$ for the scalar Laplacian $\Delta_{\zg}$ only, which contains no indicial roots
of $\cvl$.



\begin{proposition} \label{prop:main}  Let $n \geq 3$, and let $\lambda_1$ denote the
smallest nonzero eigenvalue of the scalar Laplacian $\Delta_{\zg}$. Then the only indicial roots
of $\cvl$ on $\RR \times N$ in the region
\begin{equation*}
\Big\{ |\Im (\lambda)| \le \frac{4 (n-2)^2}{ n-1} |\lambda|\Big\} \cup
\Big\{|\Im (\lambda)  | < \frac{{\sqrt{\lambda_1}}}{4(n-2)}, |\lambda|< \sqrt{\frac{\lambda_1}{2(n-1)}} \Big\}
\end{equation*}
are either $\lambda = 0$ or else lie on the imaginary axis.   In any case, there is a horizontal
strip $\{|\Im( \lambda)| < \eta\}$ which contains only the indicial root $\lambda = 0$.
\end{proposition}

The excluded region, pictured in Figure~\ref{F5II11.1}, is the union of a sector containing the
positive and negative real axes and the region inside a disc and between two horizontal lines.
\begin{figure}
\begin{center}
\includegraphics[height=5cm]{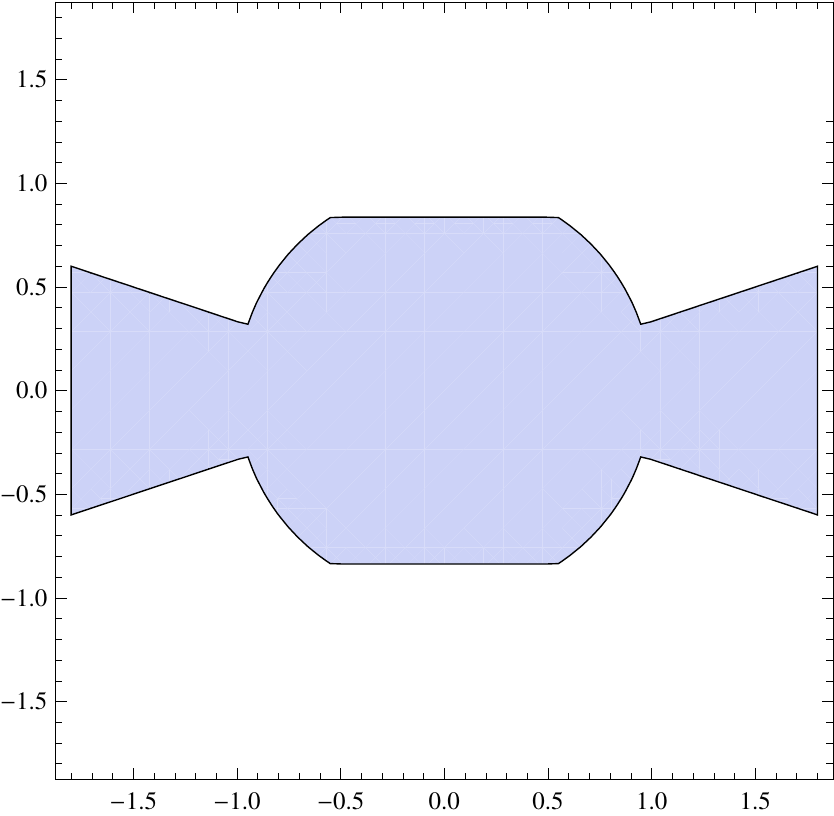}
\end{center}
\caption{\label{F5II11.1}There are no indicial roots in the shaded region except at the origin.}
\end{figure}

\medskip

\proof  We first establish the much simpler fact that the only indicial root on the real line is $0$.

First, if $\lambda = 0$, then $\Delta_{\zg} f = 0$ and $\widetilde{\cvl} Y = 0$, so $f$ is constant
and by \eqref{redquadcvl}, $\tilde{C}(Y) = 0$ and $\delta^{\zg} Y = 0$. Thus $Y$ is a Killing vector on $N$, if any.

To proceed further, we establish some identities satisfied by solutions of $I_{\lambda}(\hat{f},\hat{Y})=0$.
For convenience henceforth, we omit the hats from $f$ and $Y$, and also omit the superscript $\zg$ from $\delta$.

We first take the $L^2(N)$ inner product of the first equation in \eqref{indfam} and integrate by parts to get
\begin{equation}
||\nablah f||^2 + 2\left(1-\frac{1}{n}\right) \lambda^2 ||f||^2 + i\lambda \left( 1 - \frac{2}{n}\right) \langle \delta Y, f \rangle = 0,
\label{1cir}
\end{equation}
and
\begin{equation}
\frac12 || \tilde{C}(Y)||^2 + \frac{2}{n^2} ||\delta Y||^2 + \lambda^2 ||Y||^2 + i\lambda \left( \frac{2}{n} - 1\right) \langle
\nablah f, Y \rangle = 0.
\label{2cir}
\end{equation}
In this second equation we used \eqref{redquadcvl}.

Next, recalling that we are using the sesquilinear inner product which is complex antilinear in the second factor, integration-by-parts leads to the following three identities:
\begin{equation}
|| \nablah f||^2 + |\lambda|^2 ||Y||^2 - || \nablah f + i \bar{\lambda} Y ||^2- i\lambda \langle \nablah f, Y \rangle
+ i \bar{\lambda} \langle Y, \nablah f \rangle = 0,
\label{3cir}
\end{equation}
\begin{equation}
\begin{split}
\frac12 ||\tilde{C}(Y)||^2 & + \frac{2}{n}\left(1 - \frac{1}{n} \right) |\lambda|^2 ||f||^2 -
\frac12 \left\| \tilde{C}(Y) - \frac{2i\lambda}{n} f \zg \right\|^2 \\ & - \frac{2i\bar{\lambda}}{n^2} \langle \delta Y, f \rangle +
\frac{2i\lambda}{n^2} \langle f , \delta Y \rangle = 0,
\end{split}
\label{4cir}
\end{equation}
and
\begin{equation}
\begin{split}
\frac{2}{n^2} & ||\delta Y||^2 + 2|\lambda|^2 \left( 1 - \frac{1}{n}\right)^2 ||f||^2 - 2 \left\|
i\lambda \left(1 - \frac{1}{n}\right) f + \frac{1}{n} \delta Y \right\|^2 \\
& + \frac{2i\lambda}{n} \left(1 - \frac{1}{n}\right) \langle f, \delta Y \rangle - \frac{2i\bar{\lambda}}{n} \left(1 - \frac{1}{n}\right)
\langle \delta Y, f \rangle = 0.
\end{split}
\label{5cir}
\end{equation}
For $\lambda\ne 0$ we now form the combination
$$
\frac{\overline \lambda}{{\lambda}} \eqref{1cir} + \eqref{2cir} - \eqref{3cir} - \eqref{4cir} - \eqref{5cir} = 0
$$
to get
\begin{equation}
\begin{split}
& \frac12 \left\| \tilde{C}(Y) - \frac{2i\lambda}{n} f \zg \right\|^2 + 2 \left\|i\lambda \left(1 - \frac{1}{n}\right) f +
\frac{1}{n} \delta Y \right\|^2  \\
& + || \nablah f + i \bar{\lambda} Y ||^2 +
\left( \frac{\bar{\lambda}}{\lambda} - 1 \right) ||\nablah f||^2 + (\lambda^2 - |\lambda|^2) ||Y||^2 = 0.
\end{split}
\label{identity}
\end{equation}
For simplicity below, we write this as
\begin{equation}
\calP + \calQ + \calR + \left( \frac{\bar{\lambda}}{\lambda} - 1 \right) ||\nablah f||^2 + (\lambda^2 - |\lambda|^2) ||Y||^2 = 0,
\label{identitymod}
\end{equation}
where the first three terms here correspond to the first three terms of \eqref{identity}, in their respective order.

Suppose now that $\lambda \in \RR \setminus \{0\}$. Then $\bar{\lambda}/\lambda = 1$
and $\lambda^2 = |\lambda|^2$, so we deduce from \eqref{identity} that $\nablah f = -i \bar{\lambda} Y$
and $\delta Y = -i\lambda (n-1) f$. Together these give $\Delta_\zg f + (n-1)|\lambda|^2 f = 0$,
whence $f = 0$ since $\Delta_\zg$ is nonnegative. This implies, in turn, that $Y = 0$. Hence
there are no nonzero real indicial roots.

Now we proceed to study indicial roots off the real line.  As a first step, we note two identities,
which are {\it only valid when $\mbox{arg}\, \lambda \neq \pi/2, 3\pi/2$}. The first, obtained by
taking the imaginary part of \eqref{identity}, states that
\begin{equation}
||\nablah f ||^2  = |\lambda|^2 ||Y||^2.
\label{id1}
\end{equation}
(The reason this does not hold when $\lambda \in i \RR$ is that this imaginary part has an overall
factor of $\sin (2\arg \lambda)$.)
Next, take the real part of \eqref{identity} and substitute \eqref{id1} to get
\begin{equation}
\calP + \calQ + \calR = 4 (\Im( \lambda))^2 ||Y||^2.
\label{id2}
\end{equation}





If $\lambda \in \CC \setminus (\RR \cup i\RR)$, then it is clear from \eqref{id1} and
\eqref{indfam} that if $(f,Y)$ is a nontrivial solution to $I_\lambda(\cvl)(f,Y) = (0,0)$ then
both $f$ and $Y$ are nontrivial.  Multiplying this solution by a constant, we assume that
\[
||Y|| = 1 \Rightarrow ||\nablah f|| = |\lambda|.
\]

On a compact manifold, integration by parts shows  that the first non-zero eigenvalue of the Laplacian is non-negative. An identical calculation applies to   $\widetilde{\cvl}$, leading to the same conclusion, compare \eq{redquadcvl}.
We note the following:

\begin{lemma}
Suppose $I_\lambda(\cvl)(f,Y) = 0$ with $\lambda \neq 0$. Then, if $\lambda_1>0$ and $\nu_1>0$ are the lowest
nonzero eigenvalues of $\Delta_{\zg}$ and $\widetilde{\cvl}$, we have that
\begin{eqnarray*}
 \|f\|^2 & \leq & \frac {1}{\lambda_1} \|\nablah f\|^2,  \ \mbox{and} \\
 \|Y\|^2 & \leq &  \frac {1}{\nu_1} \left( \frac{1}{2}\|\wL Y \|^2 + \frac{2}{n^2}\|\delta Y\|^2 \right).
\end{eqnarray*}
\label{eigenlemma}
\end{lemma}


\proof On the compact manifold $N$ we have the decomposition: $f= \sum_{k=0}^{\infty} f_k$
where $f_k$ is an eigenfunction of $\Delta_\zg$ associated with
the eigenvalue $\lambda_k$ and where $\{\lambda_i\}_{i \geq 0}$
is a strictly increasing sequence with $\lambda_0 = 0$. The
kernel of $\Delta_\zg$ is the space of  constant functions;
integrating the first of  (\ref{indfam}) with $\hat w =0$ one gets $\langle f, 1
\rangle = 0$, whence $f_0=0$. Then:
$$\|f\|^2 = \sum_{k=1}^{\infty} \|f_k\|^2  \leq \frac{1}{|\lambda_1|}
\sum_{k=1}^{\infty}|\lambda_k| \|f_k\|^2 = -\frac{1}{|\lambda_1|}
\langle \Delta_{zg} f, f\rangle = \frac{1}{|\lambda_1|}\|\mcD  f\|^2
 \;.
$$

The second inequality is proved in the same way, after checking
that for all $Z$ such that $\widetilde{\cvl }Z =0$ we have
$\langle Y, Z\rangle = 0$; this is done by multiplying with $Z$ the second of  (\ref{indfam}), after setting $\hat W =0$ there, integrating over $\spOmega $, and
using the fact that $\mdiv_\zg Z=0$ for $Z$ in the kernel of
$\wL$. \qed


\medskip

Returning to the proof of Proposition~\ref{prop:main},  define
$$
\psi:=\mimath  \lambda(1-\frac{1}{n})f + \frac{1}{n}\delta Y,
$$
so that, by \eqref{id2},
\bel{psinorm}
\|\psi\| \le 2 |\Im(\lambda)|.
\ee
Inserting $\delta Y = n\psi - i\lambda (n-1) f$ into \eqref{1cir}, and recalling \eqref{id1}  gives, after some simplification,
\begin{equation}
0 = |\lambda|^2 + \lambda^2 (n-1) ||f||^2 + i\lambda (n-2) \langle \psi, f \rangle = 0.
\label{idlam}
\end{equation}
%
Dividing by $|\lambda|^2$ we find that
\bel{strangeid}
 -1 = (n-1)\frac{\lambda^2}{|\lambda|^2} \|f\|^2  +
\mimath  \frac{\lambda }{|\lambda|^2 } (n-2) \langle  \psi, f \rangle.
\ee

We claim that this equation has no solutions when $\lambda$ is sufficiently small. Indeed, let us denote
the two terms on the right by $J_1$ and $J_2$. Then, using Lemma~\ref{eigenlemma} and \eqref{psinorm} gives
\[
|J_2| \leq 2 (n-2) \frac{|\Im (\lambda)|}{ |\lambda|}  \frac{||\nablah f||}{\sqrt{\lambda_1}} =
2 (n-2) \frac{|\Im( \lambda)|}{\sqrt{\lambda_1}}.
\]
%
Hence the last term in \eq{strangeid} will have norm smaller than $1/2$ if
$$
|\Im (\lambda)  | < \frac{{\sqrt{\lambda_1}}}{4(n-2)}
 \;.
$$
Similarly, $|J_1| < 1/2$ provided
$$
 {|\lambda|^2} <  \frac{\lambda_1}{2(n-1)}
 \;.
$$
We conclude that there are no indicial roots in the intersection of these regions, which is the intersection
of a ball with a slab.

We turn to showing that $\lambda$ cannot lie in the sectorial part of this region. Suppose that
$\arg \lambda \in (-\pi/4, \pi/4) \cup (3\pi/4, 5\pi/4)$.  Then $J_1$ has positive real part,
while $|J_2| \leq 2(n-2)\|f\|$.
 Thus, for
$$
\|f\|\le \frac 1 {2 (n-2)}
$$
the equality \eq{strangeid} cannot be true, since $J_2$ must have modulus greater than one to
compensate for the positive real part of $J_1$.

Next, note that the three points $(-1, 0, J_1)$ form a triangle in the complex plane which,
for $\lambda$ with argument as above, has angle larger than $\pi/2$ at $0$. This implies that
$|1+J_1| > |J_1| = (n-1)\|f\|^2$.  Thus
\[
|J_2| = \frac{|J_2|}{|J_1|} |J_1| < \frac{|J_2|}{|J_1|} |1 + J_1|,
\]
so \eqref{strangeid} is impossible if $|J_2| / |J_1| < 1$. We have already ruled out the
possibility that $||f|| \leq \frac{1}{2(n-2)}$, so we can assume the opposite. Hence
\[
\frac{|J_2|}{|J_1|} \leq \frac{2(n-2) (|\Im( \lambda)|/|\lambda|)||f||}{(n-1)||f||^2}
< \frac{4(n-2)^2}{n-1} \frac{|\Im( \lambda)|}{|\lambda|},
\]
and this is less than or equal to one if $|\Im( \lambda)|/|\lambda| \le (n-1)/4(n-2)^2$.
%

%

Finally, the estimates above do not give information about the indicial roots on the imaginary
axis.  As we have already described, it is known that the set of all indicial roots is discrete
in the plane, so from this general result there is necessarily a strip $|\Im( \lambda)| < \eta$
such that the only indicial root in it is $\lambda = 0$. This is sufficient for the mapping properties
we describe later, but is not particularly satisfactory given the explicit nature of the other bounds above.

This proves Proposition~\ref{prop:main}.
\qed

\section{The conformal vector Laplacian  on $\mathbb Z$-periodic cylinders}
Let us now consider a space $X$ which is the cyclic cover of a compact manifold $\zX$.
This means that the map $X \to \zX$ induces a surjection in fundamental group with kernel
isomorphic to $\mathbb Z$. Choose a generator $\Gamma$ for the group of deck transformations
on $X$; this can be thought of as a translation, and $X$ itself as being cylindrical,
although it need not be homeomorphic to a product $\RR \times N$. If $\zX$ carries a Riemannian
metric $\zg$ and a TT tensor $\zK$, then we can lift these to a metric $g$ and TT tensor $K$ on $X$,
and $(X,g,K)$ is an initial data set if and only if $(\zX, \zg, \zK)$ is.

There are already some interesting examples when $X$ is a product $\RR \times N$ and $g$ is periodic.
Indeed, the paper \cite{PTCMazzeoCylindrical} classifies the ``$1$-dimensional'' solutions of the
Lichnerowicz equation in this setting, i.e.\ the solutions on a product cylinder which depend only
on $x \in \RR$, at least assuming that $||K||_g$ is constant. There is an interesting family of periodic
solutions in this case where the metric $g = w^{\frac{4}{n-2}}( dx^2 + h)$ is conformal to the product metric
on $\RR \times (N,h)$, and the conformal factor $w(x)$ is periodic. In the special case $K \equiv 0$,
these are the well-known Delaunay solutions for the Yamabe equation; when $K \neq 0$, they are
a new family of deformations of these which we call the constraint Delaunay solutions. There are
many other one-dimensional solutions one might consider, for example when $||K||_g$ is periodic;
these more general solutions were not studied in \cite{PTCMazzeoCylindrical}, but it
would certainly be interesting to do so.

In any case, we now describe some linear analysis describing the mapping properties of the conformal vector
Laplacian  $\cvl$ for any such periodic manifold. These results will be used in \S 6 to find a rich class of
TT tensors on manifolds with asymptotically periodic ends. As in the cylindrical setting, the emphasis is
on developing a criterion for determining when $\cvl$ is Fredholm on a given weighted Sobolev space.
There is an analogue of the notion of indicial roots in this setting which determines the allowable weight
parameters. This material is somewhat less standard than the corresponding theory for cylinders, so we
describe it more carefully.  The discussion below is drawn from \cite{MPU1}.

Let $(X, g)$ be a $\mathbb Z$-periodic manifold, and $\Gamma$ the generator for the deck transformations
of the covering $X \to \zX$, as described above.  Choose a fundamental domain $F$ for this action which is
a smooth compact manifold with two boundary components, $\del_-F \cup \del_+ F$, where $\Gamma$ induces
a diffeomorphism between $F_-$ and $F_+$. The basic example to keep in mind is the cylinder $X = \RR \times N$,
where $F = [0,T] \times N$ and $\Gamma(x,y) = (x+T, y)$, $y \in N$.

We begin with the direct integral decomposition
\begin{equation}
L^2(X) = \int^{\oplus} L^2(F)_\theta\, d\theta,
\label{dirint}
\end{equation}
described in \cite[p.~290]{ReedSimonVol4}, where $L^2(F)_\theta$ consists of the space of $L^2$ functions $u$ on $F$
which satisfy $u(\Gamma(z)) = e^{i\theta} u(z)$ almost everywhere on $X$. We explain this theory for functions,
but it extends immediately to sections of tensor bundles. The equivalence of function spaces \eqref{dirint} is
defined through the Fourier-Laplace transform on $X$,
\[
u(z) \longmapsto \hat{u}(\hat{z},\theta) = \sum_{k=-\infty}^\infty e^{-ik\theta} u( \Gamma^k(z)).
\]
This is initially only defined for smooth functions $u$ which are rapidly decreasing on $X$, but is extended to
all of $L^2(X)$ using the Plancherel formula
\[
\int_X |u(z)|^2\, dV_g = \int_{0}^{2\pi} \int_{F} |\hat{u}(\hat{z}, \theta)|^2\, dV_{\zg}\, d\theta, \quad \hat{z} \in F.
\]
In fact, this Plancherel formula shows that \eqref{dirint} is an isometric equivalence.
The inverse Fourier-Laplace transform is given by
\[
u(z) = \frac{1}{2\pi i} \int_0^{2\pi} e^{i k \theta} \hat{u}(\hat{z} , \theta)\, d\theta,
\]
where $k \in \mathbb Z$ is determined by the fact that $z$ lies in the $k^{\mathrm{th}}$ translate of the fixed fundamental
domain $F$, $z \in \Gamma^k (F)$. Note that $\hat{u}(\hat{z}, \theta)$ satisfies $\hat{u}( \Gamma(\hat{z}), \theta)
= e^{i\theta} \hat{u} (\hat{z}, \theta)$.  In other words, we may regard $\hat{u}$ as a section of a flat line bundle
$V_\theta$ over $\zX$ which has holonomy $e^{i\theta}$; the $L^2$ sections of this bundle are precisely the elements
of $L^2(F)_{\theta}$. One further important observation is that the flat connection on this bundle is unitary
if and only if $\theta \in \RR$.

All we have done here is to recast in geometric language the classical Bloch-Floquet wave theory, as described
in \cite{ReedSimonVol4}.  Its use in geometry was initiated by Taubes \cite{TaubesPeriodic}, and developed
further in \cite{MPU1}.

If $\calL$ is an elliptic operator on $X$ which commutes with $\Gamma$, then it induces an operator $\calL_\theta$
acting on $\calC^\infty(\zX,V_\theta)$ for any $\theta \in \CC$. We suppose just to be definite that the degree
of $\calL$ is $2$.  If $\calL$ is symmetric, then for $\theta \in \RR$,
$\calL_\theta$ uniquely determines a self-adjoint operator on $L^2(\zX, V_\theta)$, which thus has discrete (and real)
spectrum $\{\lambda_j(\theta)\}$, with $\lambda_j(\theta) \nearrow \infty$ as $j \to \infty$. These functions
satisfy $\lambda_j(\theta) = \lambda_j(2\pi - \theta)$ provided $\overline{\calL u} = \calL \overline{u}$.
All of this can be used to prove that the spectrum of $\calL$ on  $L^2(X)$ consists of the union of
`bands' $\cup_j \{ \lambda_j(\theta): 0 \leq \theta \leq \pi\}$.

For functions $u$ which decay sufficiently rapidly on $X$ we can define the Fourier-Laplace transform also for
complex values of $\theta$. In particular, if $u \in \calC^\infty_0(X)$, then $\hat{u}(\hat{x}, \theta)$ extends to be
an entire function of $\theta$ and the identity $\hat{u}(\hat{z}, \theta + 2\pi) = \hat{u}(\hat{z}, \theta)$ holds
for all $\theta \in \mathbb C$.  This defines $\calL_\theta$ as a holomorphic family of Fredholm operators,
and it is then a standard result in functional analysis that either $\calL_\theta$ is never invertible for any $\theta \in \CC$
or else $\calL_\theta^{-1}$ is a meromorphic family of operators, and the coefficients of the singular terms in the
Laurent expansion at each pole are all finite rank operators.  Note that to be in the second case, it
is necessary that $\calL_\theta$ have index zero for every $\theta$. We now say that $\theta_0$ is an indicial root of $\calL$
provided $\calL_{\theta_0}$ is not invertible, or equivalently, if $\calL_\theta^{-1}$ has a pole at $\theta_0$.
Since the index is zero, $\theta_0$ is an indicial root if and only if there exists a solution of
$\calL_{\theta_0} \hat{\phi} = 0$. Thus $\hat{\phi}$ corresponds to a function $\phi$ on $X$ which satisfies
$\phi(\Gamma(z)) = e^{i\theta_0} \phi(z)$, and if $\mbox{Im}\, \theta_0 \neq 0$, then this solution grows exponentially
in one direction and decays exponentially in the other on $X$. Thus, as in the cylindrical setting, indicial roots determine
the precise rates of exponential growth or decay of solutions of the homogeneous equation $\calL u = 0$, and these in turn
dictate the weighted Sobolev spaces on which $\calL$ is Fredholm.

We now state the mapping properties on weighted Sobolev spaces. First, let $H^s_g(X)$ denote the usual Sobolev
space relative to the periodic metric $g$ on $X$. Denote by $\rho \in \calC^\infty(X)$ a function on $X$ which satisfies
$\rho(\Gamma(z)) = \rho(z) + k$, normalized so that $\rho(z)$ vanishes at some point in the given
fundamental domain $F$. We then define the weighted Sobolev spaces $e^{\delta \rho} H^s_g(X)$ for any $\delta \in \RR$.
The following general result is proved in \cite{MPU1}:
\begin{prop}
With all notation as above, the mapping
\begin{equation}
\calL: e^{\delta \rho} H^{s+2}(X) \longrightarrow e^{\delta \rho} H^s(X)
\label{mapper}
\end{equation}
is Fredholm for all $s \in \RR$ if and only if $\delta \neq \mathrm{Im}\, \theta_j$ where $\theta_j$ is an indicial
root of $\calL$.
\label{mainthmperiodic}
\end{prop}

Since the set of indicial roots of $\calL$ is invariant when translated by any integer multiple of
$2\pi$, we may restrict attention to only those roots in the vertical strip $S := \{\theta \in \CC: -\pi
\leq \mbox{Re}\, \theta < \pi\}$. Since the set of indicial roots is discrete in $\CC$ and
$2\pi$-periodic, it follows that the set of imaginary parts of these roots is discrete in $\RR$, from which we obtain the
\begin{Corollary}
The mapping \eqref{mapper} is Fredholm for every $\delta \in \RR \setminus \Lambda$, where $\Lambda$ is the discrete
set of imaginary parts of indicial roots. In particular, there is a value $\delta_* > 0$ so that \eqref{mapper} is Fredholm
provided $0 < |\delta| < \delta_*$, and if there are no indicial roots on the real line, then \eqref{mapper} is Fredholm
for every $\delta$ with $|\delta| < \delta_*$.
\end{Corollary}

We now specialize this analysis to the conformal vector Laplacian $\cvl$. The specific issue we wish to address
is the existence of indicial roots on the real line.
\begin{lemma}
Suppose that the vector field $Y$ on the periodic manifold $X$ lies in $L^2(F)_\theta$, i.e.\ $\Gamma_* Y = e^{i\theta} Y$, and
satisfies the equation $\cvl Y = 0$. If $\theta \in \RR$, then $Y$ is a conformal Killing field, $C(Y) = 0$.
\end{lemma}
\begin{proof}
We have already noted that when $\theta \in \RR$, we may regard elements of $L^2(F)_\theta$ as sections of
a flat vector bundle over $\zX$, and that this flat connection is unitary. In more concrete terms, this means that
for (vector and tensor valued) sections of $V_\theta$, the integration by parts formula
\[
\int_{\zX} D_i A^{ij} \overline{Y}_j\, dV_g = - \int_{\zX} A^{ij} D_i \overline{Y}_j\, dV_g
\]
is still valid if both $A$ and $Y$ transform by $e^{i\theta}$ when pushed forward by $\Gamma$. The reason is that if
we think of this as an integration on the fundamental domain $F$, then after identifying the two boundary components
of $F$, the boundary term is
\[
\left(e^{i(\theta - \bar{\theta})}-1\right) \int_{\del_r F}  A^{ij} \overline{Y}_j \nu_i\, d\sigma_g,
\]
where $\nu$ is the unit normal to $\del_r F$ and $d\sigma_g$ is the volume measure on this boundary.
This vanishes if and only if $\theta \in \RR$. Hence, assuming this is the case, then we can use
\eqref{factorvl} to deduce that $\langle \cvl Y, Y \rangle = ||C(Y)||^2$, so that $C(Y) = 0$ as claimed.  \qed
\end{proof}

Since $C$ is a real operator, we can decompose $Y$ into its real and imaginary parts and hence deduce that
if $Y \in L^2(F)_\theta$ with $\theta \in \RR$ solves $\cvl Y = 0$, then there exists a real-valued, bounded
conformal Killing field on the periodic manifold $X$.  We summarize this in the
\begin{Corollary}
The operator $\cvl$ on the periodic manifold $(X,g)$ has an indicial root at $\theta_0 \in \RR$ if and only
if there exists a complex-valued conformal Killing vector field $Y$ on $X$ which transforms by
$\Gamma_* Y = e^{i\theta} Y$. In particular, if $X$ admits no bounded conformal Killing fields, then
$\cvl$ has no real indicial roots.
\end{Corollary}




\section{Global mapping properties}
 \label{s12III12.1}
We now describe results how the theory and results described in the previous two sections can be applied to construct
solutions of the vector constraint equation, in the CMC setting, on manifolds with a finite number of asymptotically cylindrical and
asymptotically periodic ends. We draw on the theory developed in two papers
\cite{RafeEdgeCPDE} and \cite{MPU1}, but do so partly for convenience since the results there are in precisely the form
that we need here, and in sufficiently general form.  There are results by many other authors which can be adapted to what
we need, though none are sufficiently general form.  We mention in particular \cite{LockhartMcOwenPisa,MelroseMendoza} in
the asymptotically cylindrical setting and \cite{TaubesPeriodic} for the asymptotically periodic case.

Let $(M,g)$ be a complete Riemannian manifold with ends $E_\ell$, $\ell = 1, \ldots, N$, where each $E_\ell$ (with the
induced metric) is asymptotic to either a cylindrical metric or a periodic metric. We refer to \cite{PTCMazzeoCylindrical}
for a precise description of the decay conditions, but briefly, we require that on each end $g$ differs from a metric which
is exactly cylindrical or periodic by a tensor $h$ which decays exponentially along with some number of its derivatives.
The papers \cite{RafeEdgeCPDE} and \cite{MPU1} assume that $h$ has a complete expansion along each end
in (negative) powers of the exponential of the distance function, but it is straightforward to adapt those arguments
to handle the case of metrics with finite regularity and a given rate of exponential decay toward the asymptotic limit.

We have already defined the weighted Sobolev spaces $e^{\delta \rho} H^s_g(M)$ for manifolds which
are exactly cylindrical or periodic, and where $\rho$ is a distance function. This generalizes immediately
to the manifold $(M,g)$. We let $\rho$ be a smooth function which equals $0$ on some large compact
set of $M$ and which agrees with the distance function on each cylindrical or periodic end.

We define the set of indicial roots of $\cvl$ on the end $E_j$ to equal the set of indicial roots of
the model (cylindrical or periodic) operator on that end. We then explain in Appendix C the fundamental
result that if $\delta$ is sufficiently close to $0$ and not equal to the indicial root of $\cvl$
on any end, then
\[
\cvl: e^{\delta \rho} H^{s+2}_g(M) \longrightarrow e^{\delta \rho}H^s_g(M)
\]
is Fredholm. We now make this more precise.

For each end $E_i$, define the space $\mycal Y(E_i)$ to consist of all conformal Killing fields on the asymptotic
exactly cylindrical or periodic model for $E_i$ which are globally bounded and, in the periodic case, spanned
by (the real parts of) vector fields which transform by $Y \mapsto e^{i\theta}Y$ over a period domain. Then choose
a smooth cutoff function $\chi_i$ which equals $1$ on $E_i$ and vanishes on the other ends and
set $\mathcal Y_i = \{ \chi_i Y: Y \in \mycal Y(E_i) \}$, and finally define $\mycal Y = \oplus \mycal Y_i$.

\begin{theorem} \label{thm:cylindre2}
Let $(M^n,g)$, $n \geq 3$, be a Riemannian manifold with $N$ ends, each of which is either asymptotically
cylindrical or asymptotically periodic. Suppose furthermore that there is no nontrivial globally defined conformal
Killing vector field $Y$ on $M$ which is in $L^2(M)$. Then there exists a number $\delta_*>0$ such that
if $0 < \delta < \delta_*$, then
\bel{surj}
\cvl : e^{\delta \rho} H^{k+2}_g(M;TM ) \to e^{\delta \rho} H_{g}^{k}(M;TM)
\ee
is surjective, while if $-\delta_* < -\delta < 0$, then
\bel{inj}
\cvl : e^{-\delta \rho} H_{g}^{k+2}(M ) \to e^{-\delta \rho} H_{g}^{k}(M)
\ee
is injective.  Moreover, if $-\delta_* < -\delta < 0$, then for every $k \geq 0$,
\begin{equation}
 \label{complem}
\cvl: e^{-\delta \rho} H_{g}^{k+2}(M;TM) \oplus {\mycal Y} \to e^{-\delta \rho}H^k_{g}(M;TM)
\end{equation}
is surjective, with finite dimensional nullspace.
\end{theorem}
\begin{remark}
This theorem states that if $M$ has no $L^2$ conformal Killing fields, then for any $W \in e^{-\delta \rho} H^{k}_{g}(M)$ there exists
a vector field $Y$ which decomposes as $Y = \mathring Y + Y'$ for some $\mathring Y \in {\mycal Y}$ and $Y' \in e^{-\delta \rho} H^{k+2}_{g}(M)$
and which satisfies $\cvl Y = W$.  Since $\mathring Y$ is conformal Killing, an immediate consequence is that this
solution $Y$ also satisfies
\begin{equation}
C(Y) \in e^{-\delta \rho} H^{k+1}_{g}(M; S^2_0(T^*M)).
\label{ckdecay}
\end{equation}

We also observe that if $E_i$ is an asymptotically cylindrical periodic end, then the vector field $\partial_x$ always
lies in $\mycal Y(E_i)$. On the other hand, a generic asymptotically periodic end has no asymptotic conformal
Killing vector fields. This implies that if all the ends of $(M,g)$ are asymptotically periodic, then for generic choices of $g$,
$\cvl: e^{\delta \rho} H^{k+2}_g(M) \to e^{\delta \rho}H^k_g(M)$  is an isomorphism for every $-\delta_* < \delta<\delta_*$.
\end{remark}

\begin{proof}
We first assume that the difference between $g$ and the its asymptotically cylindrical or periodic model along each end has an
asymptotic expansion in (not necessarily integer) powers of $e^{-\rho}$. This assumption is merely so that our hypotheses fit
within the framework of the results we quote, but is easy to remove, as we describe at the end of the proof.

There is a general principle that states in this setting that if the conformal vector Laplacian $\cvl$ associated
to the model cylindrical or periodic metric on each end $E_j$ is invertible on a given weighted space, with
inverse $G_j$, then we may patch together these inverses to get a parametrix $G$ for $\cvl$ on all of $M$.
This parametrix is an approximate inverse in the sense that $\cvl G = \mbox{Id} - K$ and $G \cvl = \mbox{Id} - K'$,
where $K$ and $K'$ are compact operators acting between the appropriate weighted Sobolev spaces.
Slightly more generally, it is sufficient to produce an operator $G_j$ on each end such that
$\cvl G_j = \mbox{Id} - K_j$, $G_j \cvl = \mbox{Id} - K_j'$ where $G_j$ and $K_j$, $K_j'$ are bounded
between the appropriate weighted spaces, and then the global parametrix $G$ can be obtained by
patching together these local parametrices. We recall the details of this argument in Appendix C below,
and refer there for a more precise formulation.

This principle shows that to prove the assertion in the statement of this theorem about the Fredholmness
of $\cvl$ we must exhibit the local parametrices for the model operators on each end.  For the
asymptotically cylindrical ends, this parametrix construction is the content of \cite[Theorem 4.4]{RafeEdgeCPDE},
while in the asymptotically periodic case it is\cite[Theorem 4.8]{MPU1}. (The reader should keep in mind that the
shift of the weight $\delta$ by $1/2$ in \cite[Equation~(4.3)]{RafeEdgeCPDE} is simply a result of a different
indexing of these weighted Sobolev spaces.)  These two results provide the existence of a suitable
parametrix which is bounded if and only if the weight parameter is not the imaginary part of an indicial root.
However, we have proved in Proposition~\ref{prop:main} (cylindrical case) and Proposition \ref{mainthmperiodic}
(periodic case) that there exists $\delta_* > 0$ such that there are no indicial roots with imaginary
part in the punctured interval $(-\delta_*,  \delta_*) \setminus \{0\}$. This proves the first assertion.

If we only know that $g$ decays to the model cylindrical or periodic metric at some exponential rate $e^{-\delta'\rho}$,
along with a finite number of its derivatives, then the conformal vector Laplacian for $g$ is a sum of
the conformal vector Laplacian for an exactly cylindrical or periodic metric and a perturbation term which
has all coefficients decaying like $e^{-\delta' z}$.  To be definite, consider the action of $\cvl$ on
$e^{\delta \rho} H^{k+2}_g$ for $0 < \delta < \delta_*$. First let $G_0$ be a parametrix for the conformal
vector Laplacian for the perturbed exactly cylindrical or periodic metric. This operator is bounded
between $e^{\delta \rho} H^k_{g}$ and $e^{\delta \rho} H^{k+2}_{g}$, and $\cvl \circ G_0 = \mbox{Id} - K_0$,
where $K_0: e^{\delta \rho} H^k_{g} \to e^{\delta - \delta'} H^{k+2}_{g}$ is bounded. However, since $e^{\delta - \delta'}H^{k+2}_{g}
\hookrightarrow e^{\delta \rho} H^k_{g}$ is a compact inclusion, $G_0$ is still a parametrix even for the conformal
vector Laplacian of the original metric $g$. The same argument works if $-\delta_* < -\delta < 0$.

We next observe that if $Y$ is an $L^2$ solution $\cvl Y = 0$, then \cite[Cor. 4.19]{RafeEdgeCPDE} (cylindrical
case) and \cite[Prop. 4.14]{MPU1} (periodic case) show that $Y \in e^{-\delta_* \rho} H^k_g$ for all $k \geq 0$.
Furthermore, the usual integration by parts, which is justified by the decay of $Y$, shows that
$C(Y) = 0$, so $Y$ must be an $L^2$ conformal Killing field. Hence under the hypothesis that $(M,g)$
admits no $L^2$ conformal Killing fields, then the mapping \eqref{inj} is not just Fredholm, but actually injective.
Moreover, by a straightforward duality argument, \eqref{surj} is surjective.

Now fix any $W \in e^{-\delta \rho} H^k_g(M;TM)$, for some $-\delta_* < -\delta < 0$. By the surjectivity
statement we have just proved, there exists $Y \in e^{\delta \rho} H^{k+2}_g(M;TM)$ such that $\cvl Y = W$.
We now cite \cite[Thm. 7.14]{RafeEdgeCPDE} (cylindrical case) and \cite[Lemma 4.18]{MPU1} (periodic case),
which assert that this solution decomposes as $Y = \mathring Y + Y'$, with $\mathring Y \in \mycal Y$
and $Y' \in e^{-\delta \rho} H^k_g$.  We remark that if $W = 0$ and $Y \in L^2$, then these same
decomposition results show that $Y \in e^{-\delta \rho} H^k_g$ for any $k \geq 0$, since the
term $\mathring Y$ does not lie in $L^2$; this was the result of the last paragraph.

This completes the proof of the theorem. \qed
\end{proof}


As an immediate consequence, we obtain a result which is one of the main goals of this paper:
\begin{prop} Let $(M,g)$ be a complete Riemannian manifold with all ends either asymptotically cylindrical
or asymptotically periodic. Use all notation as above, and assume that $(M,g)$ admits no nontrivial
global $L^2$ conformal Killing fields.   Let $A \in e^{-\delta \rho}H^k_g(M; S^2_0(T^*M))$ be a decaying
trace-free symmetric two-tensor, with $\delta$ sufficiently small. Then there exists a vector field
$Y = \mathring Y + Y' \in \mycal Y \oplus e^{-\delta \rho} H^{k+1}_g(M;TM)$ such that
$A + C(Y)$ is a solution of the vector constraint equation.
\end{prop}
\begin{proof} Following the procedure outlined in the introduction, and using the mapping properties
above, we see that if $A \in e^{-\delta \rho} H^k_{g}$, then $\delta A \in e^{-\delta \rho} H^{k-1}_g$, and hence
there exists $Y = \mathring Y + Y' \in \mycal Y \oplus e^{-\delta \rho} H^{k+1}_g$ such that $\cvl Y = \delta A$.
The actual solution of the vector constraint equation is equal to $A + C(Y)$.   \qed
\end{proof}

\medskip

We conclude this section with one further application of Theorem~\ref{thm:cylindre2}. This is the
analogue of the York splitting theorem in this geometric setting.
\begin{prop}
Let $(M^n,g)$ be a complete manifold with a finite number of ends, each of asymptotically cylindrical or
asymptotically periodic type. Suppose that there exist no $L^2$ conformal Killing fields on $M$.
Then for any $0 < \delta < \delta_*$, $k \geq 0$, and each choice of sign $\pm$, there is a direct sum decomposition
\[
\begin{split}
e^{\pm \delta \rho} H^k_g(M;S^2_0(T^*M)) = \{ & C(Y): Y \in e^{\pm \delta \rho} H^{k+1}_g(M;TM)\} \\
& \oplus \{ h \in e^{\pm \delta \rho} H^k_{g}(M;S^2_0(T^*M)): \delta h = 0 \}
\;.
\end{split}
\]
\label{York}
\end{prop}
\begin{proof}
This result is simply a manifestation of the fact that $\cvl$ has closed range when acting on $e^{\pm \delta}H^{k+2}_g(M;TM)$,
as well as the factorization $\cvl = C^* \circ C = \delta \circ C$. However, we must argue slightly differently
when proving this decomposition on spaces with positive weights than with negative ones.

First consider the mapping \eqref{surj}, with weight $\delta \in (0, \delta_*)$. For any $h \in e^{\delta \rho} H^k_g(M; S^2_0(T^*M))$,
we wish to write
\[
h = C(Y) + \kappa,
\]
where $Y \in e^{\delta \rho} H^{k+1}_g(M; TM)$ and $\delta^g \kappa = \tr^g \kappa = 0$.  To find the appropriate
$Y$, take divergence of both sides of this equation to get $\cvl Y = \delta h$. Using that \eqref{surj} is surjective,
we can find a solution $Y \in e^{\delta \rho} H^{k+1}_g$, and if we then set $\kappa := h - C(Y)$, then
$\kappa$ is trace-free and has vanishing divergence. Thus $h = \kappa + C(Y)$ is the required decomposition.

Note that the vector field $Y$ is not uniquely determined since the mapping \eqref{surj} has nontrivial nullspace.
However, using facts described in the proof of Theorem~\ref{thm:cylindre2}, any element $Z \in e^{\delta \rho} H^k_g$
which satisfies $\cvl Z = 0$ satisfies $Z = \mathring Z + Z'$ where $\mathring Z \in \mycal Y$ and
$Z' \in e^{-\delta \rho}H^k_g$.  We have proved in \SS 4 and 5 that $C(\mathring Z)$ necessarily vanishes. This means
that the usual integration by parts
\[
0 = \langle \cvl Z, Z \rangle = ||C(Z)||^2
\]
is still valid, so that $C(Z)$ itself must vanish. This means that the term $C(Y)$ which appears in the decomposition
for $h$ above is actually well-defined, even though $Y$ itself is only determined up to an element $Z$ of the nullspace of $\cvl$.

Now consider the decomposition on the negatively weighted spaces. For this we use the injectivity of the map \eqref{inj}.
Indeed, if $h \in e^{-\delta \rho} H^k_g(M; S^2_0(T^*M))$, then we can find $Y \in e^{\delta \rho} H^{k+1}_g(M;TM)$ which
solves $\cvl Y = \delta h$ and $Y = \mathring Y + Y'$, where $\mathring Y \in \mycal Y$ and $Y' \in e^{-\delta \rho}H^{k+1}_g$.
Observing that $C(\mathring Y) = 0$ on the exactly cylindrical or periodic model for each end, we see that
$C(Y) \in e^{-\delta \rho} H^k_g$.  Thus we have shown that $h = C(Y) + \kappa$,  where $Y \in e^{-\delta \rho} H^{k+1}_g$,
and where $\kappa \in e^{-\delta \rho} H^k_g$ is trace-free and divergence-free. \qed
\end{proof}

\appendix

\section{Conformal Killing vectors on cylinders}
 \label{s17II11.2}
We consider the conformal Killing vector equation $C(X)=0$ for
a metric
\bel{17II11.10}
  g=dz^2+ h
  \;,
\ee
where $h$ is a Riemannian metric on an $(n-1)$-dimensional
compact manifold $\spOmega $. For the convenience of the reader
we repeat here Equations~\eq{Ceqspm}-\eq{Ceqspm3}:
\beal{Ceqspmx}
 C_{zz} &=& \frac 2 n \left(   ({n-1}) \partial_z f - \divh Y
 \right)
 \;,
 \\
 \label{Ceqspm2x}
 C_{zA} &=& \partial_A f +\partial_z Y_A
 \;,
 \\
 \label{Ceqspm3x}
 C_{AB} &=& \mcD_A Y_B + \mcD_B Y_A - \frac 2 n ( \partial_z f + \divh Y )h_{AB}
 \;.
\eea
%
It immediately follows from $C_{AB}=0$ that $Y$ is a conformal Killing vector
of $(N,h)$:
\beal{Ceqspm3x3}  \mcD_A Y_B + \mcD_B Y_A - \frac 2 {n-1} \divh
Y  h_{AB} =0
 \;.
\eea
Differentiating \eq{Ceqspm3x3} with respect to $z$ and using $C_{zz}=0$ we find
that $\psi:=\partial_z f$ satisfies
$$
  \mcD_A   \mcD_ B \psi = \Delta_{zg} \psi h_{AB}
  \;.
$$
Keeping in mind that we have assumed $N$ to be compact, we
conclude from \cite[Theorem~21]{Kuehnel} that either
\bel{6III11.1}
 \partial_A \psi\equiv \partial_A \partial_z f =0
  \;,
\ee
or $(N,h)$ is a sphere with the round metric.

Suppose, first, that $(N,h)$ is a round sphere, the metric $g$
is then conformal to the flat metric on $ \R^n\setminus \{0\}$.
Now, it is not too difficult to check that any conformal Killing vector on $ \R^n\setminus \{0\}$ extends to a conformal Killing vector on   $ \R^n $.
As is well known, in coordinates $x^i$ on $\R^n$ in which the
Euclidean metric $\delta$ is represented by the identity
matrix, such vector fields take the form
$$
 Y^i = 
\frac 12 A^i{}_{jk  }  x^j x^k   +
  A^i{}_{j  }x^j   +
 A^i
 \;,
$$
for a set of constants $A^i$, $A^i{}_{j   }$, $A^i{}_{jk  }$,
with $A^i{}_{jk  }$ symmetric in the lower indices, and
$$
 A_{(ij)k } - \frac 2 n  A^m{}_{mk } \delta_{ij} = 0 =
 A_{(ij)} - \frac 2 n  A^m{}_{m} \delta_{ij}\;.
$$
Setting  $z=\ln r$ we have
$$
 \delta = dr^2 +r^2 d\Omega^2 = r^{2} (dz^2 + d\Omega^2) = e^{2z} g
 \;,
$$
with $g$ as in \eq{17II11.10}. So if  $\sum_{i,j,k}|A^i{}_{jk
}| \not = 0$ the $g$--length of $Y$ behaves as
$$
 |Y|_g= \sqrt{ g(Y,Y)} = e^{-z} \sqrt{ \delta(Y,Y)} \sim e^{z}
 \;;
$$
Otherwise, assuming that $\sum_{i,j}|A^i{}_{j  }| \not = 0$ the
$g$--length of $Y$ behaves as
$$
 |Y|_g  \sim 1
 \;.
$$
Finally, if  $\sum_{i,j,k}|A^i{}_{jk  }| =
0=\sum_{i,j}|A^i{}_{j  }| $ but $\sum_{i }|A^i | \ne 0$ we
obtain
$$
 |Y|_g \sim O(e^{-z})
 \;.
$$
In particular, when $h$ is the round metric on a sphere  there
exist conformal Killing vectors that decay to zero
exponentially fast along the cylindrical ends. Note that such
vectors are \emph{essential  conformal Killing vectors for
$g$}, i.e. they are \emph{not} Killing vectors for $g$ (though
they are, of course, for $\delta$).

It remains to consider $(N,h)$ such that $h$ is \emph{not} a round metric on a sphere.
Then  \eq{6III11.1} holds, in particular $\partial_z f$ is constant on each level set of $z$.
Since the integral over $N$ of $\divh Y$ is zero, we conclude
from \eq{Ceqspmx} that
$$
 \partial_ z f = 0 \;.
$$
Hence $\divh Y =0$ as well, and  $Y$ is a Killing vector of $h$
for all $z$. {}From \eqref{Ceqspm2x} we obtain $\partial_z^2 Y=0,$ hence $Y=V
+ z W$ for some $z$-independent $h$-Killing vector fields $V$
and $W$. Using \eqref{Ceqspm2x} again we find that $W=\mcD
\psi$ for a function $\psi$. But then  $ \Delta_{zg} \psi=\divh W
=0$, so $\psi$ is a constant and we infer that $W$ vanishes.
Hence $Y$ is $z$-independent. From \eqref{Ceqspm2x} we deduce
that $f$ is a constant. We conclude that $X=f\partial_z +Y$,
where $f$ is a constant and $Y$ is a conformal vector field on
$(N,h)$.

As a byproduct of the analysis above, we obtain:

\begin{Proposition}
 \label{P5II12.1}
\begin{enumerate}
 \item {Cylindrical metrics
have no Killing vector fields that decay along the cylindrical
ends.}
 \item {Cylindrical metrics
have no conformal Killing vector fields that decay along the cylindrical
ends unless $(N,h)$ is a round sphere.}
\end{enumerate}
\end{Proposition}

\section{Indicial roots of the sphere $\mathbb{S}^{n-1}$}
\label{subsec:indicialsphere} In this paragraph we give a
complete description of the set of indicial roots in the case
where the manifold $\Omega$ is the sphere $ \mathbb{S}^{n-1}$
with its standard unit round metric $h$.  The Ricci tensor of
$(\mathbb{S}^{n-1},h)$ is given by  $\Ric = (n-2)h$ and this
provides an important simplification in the study of the equations
for the indicial exponents,  $A_{\lambda}(f,Y) =0$:
\begin{eqnarray}
-2 \left( 1-\frac{1}{n} \right)  \lambda^{2}f - \Delta_h f +\mimath  \lambda
\left(1- \frac{2}{n} \right) \mdiv_h Y &=&0 \;, \label{eq:trans1}\\
-\widetilde{\cvl } Y + \mimath   \lambda  \left(1-\frac{2}{n}  \right)  \mcD  f
-\lambda^{2} Y &=&0\;. \label{eq:trans2}
\end{eqnarray}
This is due to the fact that on a round sphere we have
simple commutation properties for the operators $\mdiv$,
$\grad$ and $\widetilde{\Delta_{\mathbb{L}}}$ (from now we
suppress the subscript $h$ for these differential operators as
they will always refer to $h$). We recall that the eigenvalues of $\Delta$
when considered as acting on functions on $\mathbb{S}^{n-1}$
are $\lambda_j = j(n-2+j)$ where $j$ runs from $0$ to
$\infty$. For each function $f$ we can write $f
=\sum_{j=0}^{\infty} f_j$ where $f_j$ satisfies $\Delta f_j =
\lambda_j f_j$. The anounced commutations properties are the
following ones: for all smooth vector field $Y$ on
$\mathbb{S}^{n-1}$ we  have, keeping in mind that we use the convention $\Delta=-\nabla_b \nabla^b$,
\begin{eqnarray}
\nonumber
-\Delta \left( \mdiv Y \right)  &= & \nabla_{b} \nabla^{b} \nabla_{a} Y^{a}  =
\nabla_{b} \left( \nabla_{a}  \nabla^{b}  Y^{a}  + R^{a}_{\;cja}h^{jb}Y^{c}
\right) \\
\nonumber
&=& \nabla_{a}  \nabla_{b} \left(  \nabla^{b}Y^{a} + R^{b}_{\; cba}
  \nabla^{c}Y^{a}   + R^{a}_{\; cba}\nabla^{b}Y^{c} \right)
\\
\nonumber
 &&  - \nabla_{b}
\Ric_{cj}h^{jb}Y^{c} \\
 \nonumber
&= & \nabla_{a}  \nabla_{b}  \nabla^{b}Y^{a} -(  n-2) \nabla_{b}Y^{b}
 \\
&= &-\mdiv \left(\Delta Y \right) - (n-2) \mdiv Y.
 \label{28VI.1}
\end{eqnarray}
A similar computation gives:
 \begin{equation}
 \Delta \left(\grad f\right) = \nabla \left(\Delta f \right) -(n-2)  \grad f.
 \label{28VI.2}
\end{equation}
It then follows that:
\begin{equation*}
\mdiv \left(   \widetilde{\Delta_{\mathbb{L}}} Y  \right) = 2(1-\frac{1}{n}) \Delta(\mdiv Y) -
2(n-2) \mdiv Y
 \;,
\end{equation*}
and
\begin{equation*}
\widetilde{\Delta_{\mathbb{L}}} \left(\grad f \right) = 2(1-\frac{1}{n})\nabla(\Delta f)-2(n-2) \grad f.
\end{equation*}

Let $\mu_j$ denote the eigenvalues of
$\Delta$  acting on vector fields on
$\mathbb{S}^{n-1}$. Integration by parts shows  that $\mu_j>0$, and from \eq{28VI.1}-\eq{28VI.2} we deduce that
$$
 -\mu_j =-\lambda_j+ n-2 =  (1-j)(n-2)-j^2 \in \{\ldots,-n-6, -1\}
    \;,
    \quad j\in \N^*
 \;.
$$

Let us now turn attention to the determination of the indicial roots.
Given $\lambda \neq 0$, assume that $(f,Y)$ is a non trivial solution of
$A_{\lambda}(f,Y)=0$.
Taking the divergence of equation (\ref{eq:trans2}) and combining with (\ref{eq:trans1}) we get:
\begin{equation}
- \Delta^{2}f+\left( 2\lambda^{2}- \frac{n^{2}-2n}{n-1}\right)  \Delta f
+\left( 2(n-2) \lambda^{2} - \lambda^{4}\right) f =0 .
\label{eq:subs}
\end{equation}
As $ \langle f, 1 \rangle =0$, as one can see integrating
equation (\ref{eq:trans1}), we have the decomposition $f
=\sum_{j=1}^{\infty} f_j$ where the sum starts with $j=1$.
Inserting this expression in (\ref{eq:subs}) we get that for
each $f_i$ which is not identically equal to zero we must have:
\begin{equation}
-\lambda^4 - 2\mu_j\lambda^2 - \mu_j^2 + \frac{(n-2)^2}{n-1}(-\mu_j + 1) = 0.
\label{eq:indicialroots}
\end{equation}
Conversely if $\lambda$ satisfies (\ref{eq:indicialroots}) for at least one
$j$, one gets a non trivial solution of $A_{\lambda}(f,Y)=0$ choosing
$f$ as an eigenfunction of $\Delta$ associated to $\lambda_j$ and $Y$ as $\kappa \grad f$ where
\begin{equation*}
 \kappa= \frac{-2(1-\frac{1}{n})\lambda^2 + \lambda_j}{\imath \lambda_{j} \lambda (1- \frac{2}{n})}.
\end{equation*}

We now solve (\ref{eq:indicialroots}). Since $\mu_j\ge -1$ the discriminant $
{\frac{(n-2)^2}{n-1}(-\mu_j +1)}$ is nonpositive, and
 we get:
$$\lambda^2_{+} = -\mu_j + \imath \sqrt {-\frac{(n-2)^2}{n-1}(-\mu_j +1)}.$$
and
$$
 \lambda^2_{-} = -\mu_j - \imath \sqrt {-\frac{(n-2)^2}{n-1}(-\mu_j +1)}.
$$
Using the fact that the solutions of the equation $z^2= a \pm
\imath b$ with $b \geq 0$ are $z_1= \sqrt{\frac{a}{2}+
\frac{\sqrt{a^2 + b^2}}{2}} \pm \imath \sqrt{-\frac{a}{2}+
\frac{\sqrt{a^2 + b^2}}{2}}$ and $z_2 = - z_1$ we conclude that
the non-zero indicial roots on the sphere $\mathbb{S}^{n-1}$
are the complex numbers:
\begin{eqnarray*}
 \lefteqn{
 \sqrt{-\frac{1}{2}\mu_j + \frac{1}{2}\sqrt{\mu_j^2 +
     \frac{(n-2)^2}{n-1}(\mu_j -1)} }}
 &&
\\
 &&
        \pm \imath \sqrt{\frac{1}{2}\mu_j +
   \frac{1}{2}\sqrt{\mu_j^2 + \frac{(n-2)^2}{n-1}(\mu_j -1)} }
\end{eqnarray*}
and their opposites, where $j$ runs from $1$ to $\infty$. Note
that the first indicial root is $\sqrt{-1}$ (associated with
$\mu_1 = 1$) and that there are no non-zero indicial roots such
that $\Im( \lambda) < 1$.

\section{Fredholm properties of elliptic operators for manifolds with many ends}
\label{Ap3I1.1}
\newcommand{\tu}{{\tilde u}}
\newcommand{\tphi}{{\tilde \phi}}
\newcommand{\dtphi}{\underline{{\tilde \phi}}}
\newcommand{\utphi}{\overline{{\tilde \phi}}}
\newcommand{\tpsi}{{\tilde \psi}}
\newcommand{\bphi}{{\bar \phi}}
\newcommand{\hphi}{{\hat \phi}}

\renewcommand{\calL}{\mathcal L}
\renewcommand{\del }{\partial}
In this appendix we formulate and prove an abstract result about Fredholm theory of elliptic
operators on complete manifolds with more than one end.  We begin with the following
abstract result.
\begin{lemma}
 \label{L8XII0.1}
Let $A: X \longrightarrow Y$ be a bounded linear operator between two Banach spaces. Then $A$ is
Fredholm, i.e.\ has closed range, and finite dimensional kernel and cokernel, if and only if there exists
a bounded operator $B: Y \longrightarrow X$ which satisfies
\[
B \circ A = \mbox{Id} - Q_1, \qquad A \circ B = \mbox{Id} - Q_2,
\]
where $Q_1: X\to X$ and $Q_2: Y\to Y$ are compact operators. If this is the case, then we can modify $B$
so that the remainder terms $Q_1$ and $Q_2$ are finite rank projectors onto the nullspace and
a complement to the range of $A$, respectively.
\end{lemma}

We apply this in the concrete setting of elliptic operators on complete manifolds with a finite number of ends,
and where $X$ and $Y$ are weighted Sobolev or H\"older spaces. The main application we have in mind,
of course, is to the conformal vector Laplacian  $\cvl$ on a complete Riemannian manifold $(M,g)$ which has some
combination of asymptotically cylindrical, periodic, Euclidean, conic, or hyperbolic ends. The results here
apply equally well to the linearized Lichnerowicz operator, as studied in \cite{PTCMazzeoCylindrical},
on the same class of manifolds.  The guiding principle is that while Fredholmness of an operator is a
global property, it is sufficient to check it `locally', i.e.\ on each end.  This `local Fredholmness', in turn,
is equivalent to the existence of parametrices with compact error terms on each end. (We assume we are
working with an elliptic operator, so a parametrix with compact error is available on any bounded
region of the manifold.)

To set up the notation, suppose that $(M,g)$ is a complete Riemannian manifold with ends $E_1, \ldots, E_N$.
We assume that each $E_i$ is an open set, that $\overline{E_i}$ has compact smooth boundary, and that there is
a `radial function' $r$, which is a smooth, strictly positive function which satisfies
\[
C_1 r \leq 1 + \mbox{dist}\,(\cdot, \del E_i) \leq C_2 r.
\]
To be very specific, we assume that each end $E$ has one of the following types of geometries:
\begin{itemize}
\item[a)] Asymptotically cylindrical; thus $E = [0,\infty) \times N$, where $N$ is compact, and
$g \sim dr^2 + \zg$ for some Riemannian metric $\zg$ on $N$.
\item[b)] Asymptotically periodic; here $E$ is one half of an infinite periodic cylinder, as described in \S 5,
which covers a compact manifold $\mathring{X}$, and $g$ is asymptotic to the lift of a smooth metric from
$\mathring{X}$. The function $r$ is commensurable with the number of  fundamental domains (or period lengths)
between a given point and $\del E$. As a special case, $E = [0,\infty) \times N$, $\mathring{X} = S^1 \times N$,
and $g$ is asymptotic to a metric on $E$ which has period $T$.
\item[c)] Asymptotically hyperbolic. As before, $E = (1,\infty) \times N$, with $(N,\zg)$ compact
Riemannian, and $g \sim dr^2 + e^{2r}\zg$.
\item[d)] Asymptotically conic; we assume that $E = [1,\infty) \times N$, $(N,\zg)$ compact Riemannian,
and $g \sim dr^2 + r^2 \zg$. The most important special case is when $(N,\zg)$ is the sphere $S^{n-1}$ with
its standard unit metric, in which case we say that $E$ is an asymptotically Euclidean end; however, the results
we describe here apply equally easily in this slightly more general setting.
\end{itemize}
In each of these cases, we say that $g$ is asymptotic to a model metric $\hat{g}$ provided $g - \hat{g} = k$
decays as $r \to \infty$ along with a certain number of its derivatives. We make precise in each case the precise
rate of decay needed.

We also let $E_0$ denote an open interior region of $M$, so that $\overline{E_0}$ is a compact manifold
with boundary and $\{E_i\}_{i=0}^N$ is an open cover of $M$. Choose a partition of unity $\chi_i$ subordinate
to this open cover. Thus each $\chi_i$ is smooth and nonnegative and has support in $E_i$ and $\sum_{i=0}^N
\chi_i = 1$.  We also choose smooth nonnegative functions $\widetilde{\chi}_i$ such that $\mbox{supp}\,
\widetilde{\chi}_i \subset E_i$ for all $i$ and $\widetilde{\chi}_i = 1$ on $\mbox{supp}\, \chi_i$.
Note that the $\sum \widetilde{\chi}_i \neq 1$, in general, i.e.\ these do not form a partition of
unity, but that $\widetilde{\chi}_i \chi_i = \chi_i$ for each $i$.


The two standard choices of function spaces are weighted Sobolev and weighted H\"older spaces.
We define the unweighted versions of these spaces first.  Let $(M,g)$ be any complete Riemannian
manifold with bounded geometry and injectivity radius bounded below. The last hypothesis is
not entirely necessary for this definition, but is true in all the cases of interest here, so we assume
it for simplicity.  For any ball $B_1(q)$ of unit radius around any point $q \in M$, we can define
the local Sobolev and H\"older norms
\[
\begin{split}
& ||u||_{s,g, B_1(q)} = \left(\sum_{j \leq s} \int_{B_1(q)} |\nabla^j u|^2\, dV_g\right)^{\frac12},  \\
& ||u||_{k, \alpha, g, B_1(q)} = \sum_{ j \leq k} \sup |\nabla^j u| + \sup_{q_1, q_2 \in B_1(q)}
\frac{|\nabla^k u(q_1) - \nabla^k u(q_2)|}{ \mbox{dist}_g(q_1, q_2)^\alpha}.
\end{split}
\]
Here $s, k \in \mathbb N$ and $0 < \alpha < 1$. We then define
\[
||u||_{s,g} = \sup_{q \in M} ||u||_{s,g,B_1(q)}, \qquad ||u||_{k,\alpha, g} = \sup_{q \in M} ||u||_{k,\alpha, g, B_1(q)}
\]
as norms on the spaces $H^s_g(M)$ and $\calC^{k,\alpha}_g(M)$, respectively.
This definition applies equally well if $u$ is a section of any Hermitian vector bundle over $M$ for which
there is a standard trivialization over each of the balls $B_q(q)$; this is certainly the case of we are dealing
with sections of a tensor bundle over $M$.  (We could equally easily have defined $L^p$-based Sobolev spaces for
any $p \in (1,\infty)$ and for any real $s$, and all results below have analogues in this more general setting).

For each $i =1, \ldots, N$, choose a weight function $w_i$; this is a smooth strictly positive function
on $M$ which equals $1$ away from the end $E_i$, and in each of the cases a) - d) has the following description:
when $E_i$ is asymptotically cylindrical, periodic and hyperbolic, we set $w_i = e^{-r}$;
when $E_i$ is asymptotically conic, we take $w_i = r$. Finally, if $a = (a_1, \ldots, a_N)$
is any $N$-tuple of real numbers, we write
\[
\begin{array}{rl}
& w^a \cC^{k,\alpha}_g(M) = \{u: w_1^{-a_1}\ldots w_N^{-a_N} u \in \cC^{k,\alpha}_g(M)\}, \\
& w^a H^s_g(M) = \{u: w_1^{-a_1}\ldots w_N^{-a_N} u \in H^s_g(M)\}.
\end{array}
\]
We can also localize these spaces to each end, and thus define $w_i^{a_i} \cC^{k,\alpha}_g(E_i)$ and
$w_i^{a_i} H^s_g(E_i)$. Finally, a subscript $0$ indicates the subspace of functions which vanish at
$\del E_i$; thus, for example, $\cC^{k,\alpha}_{g,0}(E_i)$ consists of all functions $u_i \in
\cC^{k,\alpha}_g(E_i)$ such that $u_i = 0$ at $\del E_i$.

Let $A$ be any natural geometric elliptic operator associated to one of the metrics above;
we have in mind that $A$ is the scalar Laplacian, the Hodge Laplacian on $k$-forms,
the trace Laplacian $\nabla^* \nabla$ acting on a tensor bundle, or the conformal vector Laplacian  $\cvl$,
however the results we describe below apply to a broad class of elliptic operators with
similar asymptotic behaviour. To be definite, we suppose that $A$ has the following form
in each of the geometries of interest:
\begin{itemize}
\item[a)] $A \sim \del_r^2 + \Delta_{\zg}$;
\item[b)] $A \sim A_{\zg}$, where $A_{\zg}$ is the lift to the asymptotically periodic end of an elliptic
operator on the compact manifold $\zX$;
\item[c)] $A \sim \del_r^2 + a \del_r + e^{-2r} \Delta_{\zg}$;
\item[d)] $A \sim \del_r^2 + a r^{-1} \del_r + r^{-2} \Delta_{\zg}$.
\end{itemize}
In each case, the notation $A \sim A_0$ indicates that the coefficients of $A - A_0$ decay to zero
at a rate $e^{-\delta r}$ in cases a), b) and c), like $r^{-\delta}$ in case d).   For simplicity we have omitted
the typical terms of order $1$ and $0$ which might appear in the operators of interest; perhaps the most
important thing to note is that these lower order terms need not decay in cases a)-c), but the
terms of order $0$ in the asymptotically conic case d) must decay like $r^{-2}$.

We have set up the notation so that
\[
A: w^a H^{s+2}_g(M) \longrightarrow w^a H^s_g(M)
\]
in cases a)-c) for any values of the weight parameter $a_i$, whereas
\[
A: w^a H^{s+2}_g(M) \longrightarrow w^{a-2} H^s_g(M)
\]
in case d). The mapping properties between weighted H\"older spaces is phrased analogously,
but for the purpose of brevity we do not state these separately.

We finally come to the main Fredholm mapping properties for each of these classes of operators.
As before, we do not state these in the most general contexts, but specialize to natural geometric
elliptic operators on ends which are asymptotic to warped products as given by the descriptions above.
To state these results, we must define the indicial roots of the operator $A$ in each case.

We have already described in \SS 4-5 the indicial roots of the conformal vector Laplacian  $\cvl$ on asymptotically cylindrical
and asymptotically periodic ends.  If the end $E$ is asymptotically hyperbolic, then $\zeta$ is an indicial root
of $A$ if $A e^{-\zeta r} = \calO( e^{-(\zeta + \epsilon) r})$ for some $\epsilon > 0$. Notice that since the tangential
derivatives in $A$ are already accompanied by the factor $e^{-2r}$, these terms do not affect the calculation
of the indicial roots; in other words, the indicial roots in the asymptotically hyperbolic case are determined by
an algebraic equation involving the coefficients of the derivatives in the normal direction. Finally, in case d),
the number $\zeta$ is an indicial root if there exists a function $\phi$ on the cross-section $N$ such that
$A (r^\zeta \phi) = \calO(r^{\zeta-2-\epsilon})$ for some $\epsilon > 0$. This corresponds to a leading order
cancellation, since for an arbitrary value of $\zeta$ and smooth function $\phi$, one always has
$A( r^\zeta \phi) = \calO(r^{\zeta-2})$. Just as in the asymptotically cylindrical case, the indicial roots and
the coefficient functions $\phi$ in this case are determined by eigendata for the induced operator
on the cross section $(N,\zg)$.

The basic result that we state below is that $A$ is Fredholm acting between weighted Sobolev or H\"older spaces
if and only if none of the weight parameters $a_i$ are equal to the inaginary part of an indicial root on the
corresponding end.  Actually, if any one of the ends is asymptotically hyperbolic, then this condition
must be modified, as we now describe.  As described in \cite{RafeEdgeCPDE} (see Theorems 5.16 and 6.1 in particular),
in order to show that $A$ is Fredholm, it is necessary that a certain model operator for $A$ at each point at infinity,
called the normal operator $N(A)$, must be an isomorphism between these same weighted spaces. In
our setting where $g \sim dr^2 + e^{2r} h$ is asymptotically hyperbolic and $A$ is the conformal vector
Laplacian, the normal operator $N(A)$ turns out simply to be equal to the conformal vector Laplacian on
hyperbolic space $\HH^n$ itself. Thus the extra condition we are imposing is that if the end $E_i$
is asymptoticaly hyperbolic, then the weight parameter $a_i$ must  be chosen so that
\[
\cvl^{\HH^n}: x^{a_i} H^2(\HH^n; T\HH^n) \longrightarrow x^{a_i} L^2(\HH^n; T\HH^n)
\]
is an isomorphism. Here we are thinking of $\HH^n$ as the upper half-space model with $x > 0$ and
$y \in \RR^{n-1}$. It turns out that there is always an allowable range of weight values for which this is
true, see \cite{Lee:fredholm}. We call this the critical weight range associated to the operator $A = \cvl$
on an asymptotically hyperbolic end.

\begin{proposition}
With all notation as above, suppose that no weight parameter $a_i$ is indicial, and in addition, if $E_j$
is an asymptotically hyperbolic end, then $a_j$ lies in the critical range described above. Then for
any $k \geq 0$,
\[
A: w^a H^{k+2}_g(M) \longrightarrow w^{a'} H^k_g(M)
\]
is a Fredholm mapping; here $a_i' = a_i$ unless $E_i$ is asymptotically conic, in which case
$a_i' = a_i-2$.
\end{proposition}
\begin{proof}
The way to prove this result in a uniform way independent of the type of geometry on each end
is to draw from the literature the existence of a parametrix, or approximate inverse modulo a
compact error, for $A$ on each of the types of ends. More specifically, we assert that for
$i = 0, \ldots, N$, there exist operators $Q_{1i}$ and $Q_{2i}$ such that
\[
\left| A \right|_{E_i} \circ B_i: \mbox{Id} - Q_{1i}, \quad B_i \circ \left. A \right|_{E_i} = \mbox{Id} - Q_{2i},
\]
where
\[
\chi_i B_i: w^{a'} H^s_g(E_i) \longrightarrow w^a H^{s+2}_g(M),
\]
and
\[
\chi_i Q_{1i}: w^{a'}H^{s}_g(E_i) \longrightarrow w^{a'-\epsilon} H^{s+1}_g(M),
\]
\[
\chi_i Q_{2i}: w^{a}H^{s+2}_g(E_i) \longrightarrow w^{a-\epsilon} H^{s+3}_g(M),
\]
are all bounded. We have included the cutoff functions $\chi_i$ on each of these
factors as a simple way to localize to each end.  The important fact here is that
\[
\begin{split}
& \chi_i Q_{1i}: w^{a'}H^{s}_g(E_i) \longrightarrow w^{a'} H^s_g(E_i), \\
&\chi_i Q_{2i}: w^{a}H^{s+2}_g(E_i) \longrightarrow w^{a}H^{s+2}_g(E_i)
\end{split}
\]
are both compact operators.

Although we have phrased this in fairly abstract operator-theoretic terms, we do need one specific
fact about the structure of these operators, which is that the $B_i$ are pseudodifferential
operators of order $-2$, hence have the following special property that for each $i$,  the commutator
\[
[ A, \widetilde{\chi}_i] B_i \chi_i
\]
is compact. In fact, $[A, \widetilde{\chi}_i]$ is a first order operator with compactly supported
coefficients which have support disjoint from the support of $\chi_i$. This means that the composition
above is actually a smoothing operator which maps functions on $M$ into $\cC^{\infty}_0(E_i)$.

{}From these local parametrices we now define the global parametrix
\[
B = \sum_{i=0}^N \widetilde{\chi}_i B_i \chi_i.
\]
We compute that
\[
\begin{array}{rcl}
A \circ B & = & \sum_{i=0}^N  \left( \widetilde{\chi}_i (\mbox{Id} - Q_{1i}) \chi_i - [A, \widetilde{\chi}_i] B_i \chi_i\right) \\
& = & \mbox{Id} - \sum_{i=0}^N \left( [A, \widetilde{\chi}_i] B_i \chi_i + \widetilde{\chi}_i Q_{1i} \chi_i\right).
\end{array}
\]
By our various hypotheses, all terms in the final sum are compact operators, and hence
\[
A \circ B = \mbox{Id} - Q_1,
\]
where $Q_1$ is compact.  A similar computation and argument gives the corresponding
conclusion for $B \circ A$.
\qed

It remains to cite the relevant places in the literature where the existence of these local parametrices are proved.
We first mention the paper \cite{RafeEdgeCPDE}, which gives a comprehensive treatment of parametrices for the class
of elliptic differential operators of `edge type'.  These naturally include cases a), c) and d), where we note
that although a second order asymptotically conic operator $A$ is not actually an edge operator, but $r^2 A$ is,
and this suffices for the purposes above. Parametrices for elliptic operators in the asymptotically periodic
case were first constructed by Taubes \cite{TaubesPeriodic}, and that construction was generalized in \cite{MPU1}
to allow for the possibility of indicial roots with real part $0$, as occurs in our applications here.
An earlier source which treats elliptic theory in both the asymptotically cylindrical and conic cases is the work of
Lockhart and McOwen \cite{LockhartMcOwenPisa}. Finally, parametrices in the
asymptotically hyperbolic case have also been constructed in \cite{Lee:fredholm}.
\end{proof}

\bigskip

{\sc Acknowledgements:} PTC wishes to thank the Institut des
Hautes \'Etudes Scientifiques, Bures sur Yvette, for
hospitality and financial support during part of work on this
paper.

\bibliographystyle{amsplain}
\bibliography{../references/bartnik,%
../references/myGR,%
../references/newbiblio,%
../references/netbiblio,%
../references/newbib,%
../references/reffile,%
../references/bibl,%
../references/Energy,%
../references/hip_bib,%
../references/dp-BAMS,%
../references/besse2%
}%
\end{document}